\documentclass[12pt,letterpaper, onecolumn, accepted=2022-09-25]{quantumarticle}
\pdfoutput=1

\usepackage[utf8]{inputenc}
\usepackage[english]{babel}
\usepackage[T1]{fontenc}

\usepackage{amsmath}
\usepackage{amssymb}
\usepackage{amsthm}
\usepackage{mathtools}
\usepackage{lipsum}
\usepackage{outlines}
\usepackage[spaces,hyphens]{url}
\usepackage[colorlinks,allcolors=blue]{hyperref} 
\usepackage[normalem]{ulem}
\usepackage{graphicx}
\usepackage{subcaption}
\usepackage[dvipsnames]{xcolor}
\usepackage{tikz}
\usepackage{quantikz}

\newtheorem{theorem}{Theorem}
\newtheorem{lemma}[theorem]{Lemma}

\newcommand{\ketbra}[2]{|#1\rangle\langle #2|}

\newcommand{\lrrb}[1]{\left( #1 \right)}

\newcommand{\lrcb}[1]{\left \{ #1 \right \}}

\newcommand{\abs}[1]{\left | #1 \right |}
\newcommand{\norm}[1]{\left \| #1 \right \|}

\newcommand{\sech}[1]{\operatorname{sech}\left ( #1 \right)}

\newcommand{\indicator}{\operatorname{1}_{S}}
\newcommand{\angstrom}{\mbox{\normalfont\AA}}

\newcommand{\R}{\mathbb{R}}
\newcommand{\C}{\mathbb{C}}
\newcommand{\Z}{\mathbb{Z}}

\newcommand{\defeq}{\coloneqq}

\renewcommand{\exp}[1]{\operatorname{exp}\lrrb{#1}}

\begin{document}

\title{State preparation boosters for early fault-tolerant quantum computation}

\author[1]{Guoming Wang}
\email{guoming.wang@zapatacomputing.com}
\author[2]{Sukin Sim}
\email{sukin.sim@gmail.com}
\author[2]{Peter D. Johnson}
\email{peter@zapatacomputing.com}
\affiliation[1]{Zapata Computing Canada Inc., 25 Adelaide St E, Suite 1500, Toronto, ON M5C 3A1, Canada}
\affiliation[2]{Zapata Computing Inc., 100 Federal Street, 20th Floor, Boston, MA 02110, USA}

\maketitle

\begin{abstract}
Quantum computing is believed to be particularly useful for the simulation of chemistry and materials, among the various applications. In recent years, there have been significant advancements in the development of near-term quantum algorithms for quantum simulation, including VQE and many of its variants. However, for such algorithms to be useful, they need to overcome several critical barriers including the inability to prepare high-quality approximations of the ground state. Current challenges to state preparation, including barren plateaus and the high-dimensionality of the optimization landscape, make state preparation through ansatz optimization unreliable. In this work, we introduce the method of ground state boosting, which uses a limited-depth quantum circuit to reliably increase the overlap with the ground state. This circuit, which we call a booster, can be used to augment an ansatz from VQE or be used as a stand-alone state preparation method. The booster converts circuit depth into ground state overlap in a controllable manner. We numerically demonstrate the capabilities of boosters by simulating the performance of a particular type of booster, namely the Gaussian booster, for preparing the ground state of $N_2$ molecular system. Beyond ground state preparation as a direct objective, many quantum algorithms, such as quantum phase estimation, rely on high-quality state preparation as a subroutine. Therefore, we foresee ground state boosting and similar methods as becoming essential algorithmic components as the field transitions into using early fault-tolerant quantum computers.
\end{abstract}

\section{Introduction}

Quantum computing holds the promise to solve previously intractable problems in quantum chemistry and materials \cite{aspuru2005simulated, cao2019quantum}.
Essential to this promise is the ability of a quantum computer to efficiently encode and process wave functions that represent the quantum system of interest.
This system 
is encoded on the quantum computer via a qubit representation of the Hamiltonian.
An important step of many quantum algorithms for quantum chemistry and materials is to prepare on the quantum computer an approximation to the ground state of this Hamiltonian \cite{aspuru2005simulated, kassal2008polynomial}.
The performance of many quantum algorithms relies on the performance of the ground state preparation step \cite{kitaev2002classical,ge2019faster,lin2021heisenberg}.
In the realm of ``far-term'' quantum algorithms, the runtime of quantum phase estimation and related techniques scales as $O(1/p^2)$, where $p$ is the quality, or fidelity, of the input ground state approximation \cite{ge2019faster,lin2021heisenberg}.
Improved methods \cite{knill2007optimal} reduce this runtime to $O(1/p)$ and recent work \cite{lin2020near} improves this runtime scaling to $O(1/\sqrt{p})$ at the cost of a $O(1/\sqrt{p})$ factor increase in circuit depth.
More ``near-term'' quantum algorithms, such as the variational quantum eigensolver (VQE), aim to estimate the ground state energy by directly preparing an approximation to the ground state; the quality of the energy estimate is limited by the quality of the ground state preparation \cite{Peruzzo2014,mcclean2016theory}.
The benefit of near-term quantum algorithms like VQE is that they have the potential to prepare good approximations to the ground state, while using relatively little circuit depth.
The relatively low circuit depth limits the accrual of error in the state preparation task.
While such techniques may serve as a good ``head start'' in approximate ground state preparation, recent numerical studies suggest that the approximations achievable with these techniques are insufficient for outperforming state-of-the-art classical methods \cite{mcclean2018barren,sim2021adaptive,bonet2021performance}.

This challenge motivates the use of further quantum circuitry, beyond heuristic parameterized quantum circuits to improve, or ``boost'', the ground state preparation subroutine in near-term quantum algorithms.
Such techniques are likely needed to achieve quantum advantage for chemistry with early fault-tolerant quantum computers.
While there do exist methods for improving the approximation of ground states \cite{poulin2009preparing, lin2020near} and excited states \cite{jensen2020quantum}, these methods are designed to be performant with an idealized quantum computer.
With early fault-tolerant quantum computers, we should design algorithms which accommodate the build-up of error during the quantum computation.

This motivates the need for quantum algorithms which reliably increase the overlap with the ground state of a target Hamiltonian using limited circuit depth.
One such method is developed in Ref.~\cite{anshu2021improved}. This approach, suited for few-body interacting Hamiltonians, gives a low-depth quantum circuit which ensures a fixed energy reduction.
However, it is not clear how to futher reduce the energy or increase the ground state overlap beyond this initial amount.
If a particular energy reduction or overlap is required, the reduction ensured by this method might not be sufficient.
Furthermore, the performance of the method is inversely related to the locality of the Hamiltonian and its degree of interaction, which limits its application to quantum chemistry Hamiltonians.
Nevertheless, the method may provide a good input to additional state preparation methods.
Another such method is the quantum imaginary time evolution (QITE) algorithm, introduced in Ref.~\cite{motta2020determining}. QITE aims to prepare an approximation to a low-temperature thermal state of a Hamiltonian.
By choosing a sufficiently low temperature, the output state is a good approximation to the ground state.
The advantage of the QITE algorithm compared to previous imaginary time evolution algorithms \cite{mcardle2019variational} is that the performance of the method does not rely on the quality of a proposed ansatz circuit.
Rather, this iterative method constructs the best approximation to each imaginary time evolution by making local tomographic estimates of the current state.
One disadvantage of the QITE algorithm is that its favorable runtime requires the system of interest to have local interactions; many systems of interest, including the electronic structure of molecules are most-conveniently represented with non-local interactions.
Another disadvantage of the QITE algorithm is that it is unclear how the method will perform when subject to error in the quantum circuit.
A worry is that error will accrue unfavorably through the iterative process.
The reliability of the method in the presence of error for large-scale calculations is still an open question.

More recently, Amaro et al. \cite{amaro2021filtering} proposed a near-term ground state filtering method. Drawing inspiration from the ground state filtering of Ref.~\cite{ge2019faster}, instead of applying the filter operator to the initial state, they emulate the action of the filter operator by using a parameterized quantum circuit (PQC) and optimizing the parameters to minimize the Euclidean distance between the PQC-prepared state and the filter-applied state. They carry this out in an iterative fashion until termination criteria are reached. They investigate several filters (function types) and demonstrate the performance with QAOA. 
While this work does use similar concepts of filtering that we make use of, their method relies on a heuristic parameterized circuit optimization.
Accordingly, this filter does not admit performance guarantees and is likely to suffer from the same reliability issues as VQE.

In this context our work addresses the following question: \emph{using limited circuit depth, how can we increase the overlap with the ground state in a way that works efficiently for all Hamiltonians and gives reliable performance?}
We develop the method of \emph{state preparation boosting} for reliably improving the overlap with a target state while using quantum circuits of limited circuit depth.
In contrast to the optimization of parameterized quantum circuits, state preparation boosters reliably convert precious quantum circuit depth into overlap with the target state.
Moreover, ground state boosters can be appended to an already-optimized VQE circuit to reliably increase its overlap with the target state.
Alternatively, one could append a state preparation booster to a parameterized quantum circuit throughout the VQE optimization.

In Section \ref{sec:booster_method} we introduce the concept of state preparation boosters and introduce several methods for designing them.
In Section \ref{sec:numerics} we analyze the performance of boosters for a small molecular system and show that the boosters reliably convert circuit depth into overlap with the ground state.
We conclude with a discussion on future work for state preparation boosters and an outlook on their use in helping to achieve quantum advantage.
Several important technical details are included in the appendix.

While preparing this paper, we came across a recent work \cite{he2021quantum} presenting the Gaussian filter method similar to a particular type of booster investigated in this work. They present a hybrid quantum-classical method that estimates the ground state energy but does not prepare the ground state. Their fully quantum method, which prepares the ground state, is well-suited for qubit-qumode quantum processors, but the authors do not address how to prepare the ground state on a qubit-based quantum computer, which is the focus of our work. Furthermore, our work extends beyond Gaussian filters or boosters; we present a general framework that can be used to design low-depth boosters for preparing ground states.

After an earlier version of this paper was posted on arXiv, Dong et al. \cite{dong2022ground} proposed two algorithms for approximately preparing the ground state of a given Hamiltonian on early fault-tolerant quantum computers. They aim to achieve a fidelity close to one in the output state, and hence the costs of their algorithms are relatively high. Precisely, one of their algorithms has query depth $\tilde{O}(1/\Delta)$ and query complexity $\tilde{O}(1/(\Delta \gamma^2))$, and the other has query depth $\tilde{O}(1/(\Delta \gamma))$ and query complexity $\tilde{O}(1/(\Delta \gamma))$, where $\Delta$ and $\gamma$ are the spectral gap and initial ground state overlap, respectively. Although these algorithms could be modified to have lower circuit depth at the cost of producing a worse approximation \footnote{Ref.~\cite{dong2022ground} uses a technique called quantum eigenvalue transformation of unitary matrices (QET-U) which enables the implementation of certain polynomials of $\cos(H/2)$, where $H$ is the Hamiltonian of interest. For ground state preparation, the polynomial approximates a threshold function so that the QET-U circuit approximately implements a projective measurement onto the low-energy and high-energy subspaces of $H$. The depths of the resulting circuits are proportional to the degree of this polynomial. One might use a different polynomial with lower degree to construct a shallower circuit which generates a worse approximation of the ground state.}, this idea is not explored in their paper.

In contrast, here we aim to improve heuristic ground state preparation methods like VQE by using limited amount of resources. Our circuit depth can be tuned based on our need and hardware constraints, and using deeper circuits leads to steady increase of the ground state overlap. Although our method can be also used to prepare an extremely-accurate approximation of the ground state, this requires quite deep circuits and is not the main focus of the paper. We leave it as future work to compare the resource costs of our scheme and the one of Ref.~\cite{dong2022ground}.

\section{Constructing state preparation boosters}
\label{sec:booster_method}
In this section, we describe the methodology of state preparation boosters for improving ground state preparation. Our basic idea is that, given any ansatz circuit for approximately preparing the ground state of a Hamiltonian $H$, we compose it with a shallow circuit that implements a function of $H$, so that the extended circuit yields a state that has larger overlap with the ground state of $H$. The appended operation, denoted by $f(H)$, suppresses the high-energy eigenstates of $H$ and hence effectively boosts the low-energy ones for any input state. Thus we refer to these operations as \emph{boosters}. 

Formally, let $H=\sum_{j=1}^D \lambda_j \ket{\lambda_j}\bra{\lambda_j}$ be a Hamiltonian with eigenvalues $\lambda_j$'s and eigenstates $\ket{\lambda_j}$'s, where $a \le \lambda_1 \le \lambda_2 \le \dots \le \lambda_D \le b$ for some $a, b \in \R$. Suppose $\ket{\psi}=\sum_{j=1}^D \mu_j \ket{\lambda_j}$ is the state produced by an ansatz circuit such that  $\mu_1 \neq 0$. We will find a function $f: \mathbb{R} \to \mathbb{C}$ such that $|f|$ decreases monotonically on the interval $[a, b]$. Then performing the operation $f(H)$ on the state $\ket{\psi}$ yields the unnormalized state 
\begin{align}
f(H)\ket{\psi} = \sum_{j=1}^D \mu_j f(\lambda_j)  \ket{\lambda_j},
\end{align}
whose normalized version is
\begin{align}
\dfrac{f(H)\ket{\psi}}{\norm{f(H)\ket{\psi}}} 
= \sum_{j=1}^D \mu_j' \ket{\lambda_j}
= \dfrac{1}{\sqrt{Z}}\sum_{j=1}^D \mu_j f(\lambda_j) \ket{\lambda_j},     
\end{align}
where $Z = \sum_{j=1}^D \abs{\mu_j f(\lambda_j)}^2$ and $\mu_j'=\mu_j f(\lambda_j) / \sqrt{Z}$. By the assumption on $f$, we have $\abs{\frac{\mu_i f(\lambda_i)}{\mu_j f(\lambda_j)}} \ge \abs{\frac{\mu_i}{\mu_j}}$, $\forall i \le j$. It follows that $\abs{\mu_j'} \ge \abs{\mu_j}$ for $j=1, 2, \dots, k$, for some $1 \le k \le D$. {In other words}, the amplitudes of the lower-energy eigenstates of $H$ in $\frac{f(H)\ket{\psi}}{\norm{f(H)\ket{\psi}}}$ are larger than their counterparts in $\ket{\psi}$. In particular, the overlap between $\frac{f(H)\ket{\psi}}{\norm{f(H)\ket{\psi}}}$ and the ground state of $H$ is larger than the one for $\ket{\psi}$. 
We note that although monotonicity can ensure some degree of performance,
a non-monotonic booster may be favorable if it can be implemented with shorter-depth quantum circuits.
For example, the ``truncated Gaussian boosters'' that we explore later are technically non-monotonic (due to the truncation), but their corresponding functions have steeply decreasing envelopes, making them nonetheless performant.

Although in principle every function $f$ that satisfies the monotonicity condition leads to a booster, the cost of implementing the operation $f(H)$ vary significantly for different $f$'s. Furthermore, in general, $f(H)$ is non-unitary and hence we cannot implement it with certainty. So we need to choose $f$ carefully to achieve a balance among three factors:
\begin{itemize}
    \item To what extent can $f(H)$ increase the overlap with the ground state (or the low-energy eigenstates) of $H$?
    \item What is the depth of the circuit for implementing $f(H)$?
    \item What is the probability of this circuit successfully implementing $f(H)$?
\end{itemize}
Before addressing this issue in Section \ref{subsec:design_booster_function}, we first describe a method to physically implement $f(H)$ once $f$ is determined. 

\subsection{Implementing the booster operation}
\label{subsec:implement_booster_operation}
Suppose the booster function $f$ is known. We implement the booster operation $f(H)$  by obtaining a Fourier approximation of $f$ and applying the linear combination of unitaries (LCU) method to it. Specifically, let $\hat{f}$ be the Fourier transform of $f$, i.e.
\begin{align}
\hat{f}(\xi) = \displaystyle\int_{-\infty}^\infty f(x) e^{-i 2\pi x \xi} dx, \label{eq:ft}\\
f(x) = \displaystyle\int_{-\infty}^\infty \hat{f}(\xi) e^{i 2\pi x \xi} d\xi. \label{eq:inv_ft} 
\end{align}
We convert the RHS of Eq.~\eqref{eq:inv_ft} into a finite sum by truncating the integral to the region $[-T, T]$ for some $T>0$ and then discretizing the integral on that region. {Specifically}, let
\begin{align}
    f_T(x) = \displaystyle\int_{-T}^T \hat{f}(\xi) e^{i 2\pi x \xi} d\xi,
    \label{eq:def_ft}
\end{align}
and for any $N \in \Z^+$, let
\begin{align}
    f_{T,N}(x) = \frac{T}{N} \sum_{j=-N}^{N-1} \hat{f}\lrrb{\xi_j} e^{i 2\pi x \xi_j},
    \label{eq:def_ftn}
\end{align}
where 
$\xi_j=(j+1/2)\frac{T}{N}$ for $j=-N, -N+1, \dots, N-1$. We would like to choose some $T$ and $N$ such that $f \approx f_T \approx f_{T,N}$. Note that
\begin{align}
\epsilon_T &\defeq \max_{x \in \R} \abs{f(x) - f_T(x)} \\
& = \max_{x \in \R} \abs{ \displaystyle\int_{-\infty}^{-T} \hat{f}(\xi) e^{i 2\pi x \xi} d\xi + \displaystyle\int_{T}^{\infty} \hat{f}(\xi) e^{i 2\pi x \xi} d\xi
}  \\
& \le \displaystyle\int_{-\infty}^{-T} |\hat{f}(\xi)| d\xi + \displaystyle\int_{T}^{\infty} |\hat{f}(\xi)| d\xi.
\label{eq:fourier_approx}
\end{align}
We will focus on the case where $|\hat{f}|$ decays exponentially on $\mathbb{R}$, i.e., there exists constant $c>1$ such that $|\hat{f}(\xi)| = O(c^{-|\xi|})$ as $\xi \to \pm \infty$. Then Eq.~\eqref{eq:fourier_approx} implies that $\epsilon_T=O(c^{-T})$ as $T \to +\infty$. As a consequence, choosing some $T=O(\log(1/\epsilon))$ ensures that $\max_{x \in \R }\abs{f(x) - f_T(x)} \le \epsilon$, for arbitrary $\epsilon >0$. Meanwhile, in Appendix \ref{app:dist_ft_ftn}, we show that when $f$ satisfies certain (weak) conditions, choosing some $N=\tilde{O}(1/\epsilon)$ ensures that $\max_{x \in [0, 1]} \abs{f_T(x) - f_{T,N}(x)} \le \epsilon$, for arbitrary $\epsilon >0$. Then for such $T$ and $N$, we have
\begin{align}
 \abs{f(x) - \frac{T}{N} \sum_{j=-N}^{N-1} \hat{f}\lrrb{\xi_j} e^{i 2\pi x \xi_j}} \le 2\epsilon, & ~\forall x \in [0, 1].  
 \label{eq:f_ftn_dist}
\end{align}
This gives a Fourier approximation of $f$ on the interval $[0, 1]$, which is sufficient for our purpose.

Now suppose we have a Fourier approximation of the booster function $f$:
\begin{align}
    f(x) \approx \sum_{j=1}^K \alpha_j e^{i t_j x}
    \label{eq:f_fourier_approx}
\end{align}
for some $\alpha_j \in \C \setminus \lrcb{0}$ and $t_j \in \R$. Then we have a Fourier approximation of the booster operation $f(H)$ as well:
\begin{align}
    f(H) \approx \sum_{j=1}^K \alpha_j e^{i t_j H}.
    \label{eq:fh_fourier_approx}
\end{align}
{That is}, $f(H)$ is close to a linear combination of the unitary operations $e^{i t_j H}$'s. This prompts us to use the LCU method to implement $f(H)$ approximately and probabilistically. Specifically, suppose we want to implement a linear operation $A=\sum_{j=1}^K \alpha_j U_j$ on a Hilbert space $\mathcal{H}$, where $\alpha_j \in \C \setminus \lrcb{0}$ and $U_j$ is a unitary operation on $\mathcal{H}$ for each $j$. The LCU method requires a $K$-dimensional ancilla and a unitary operation $V$ on the ancilla such that 
\begin{align}
V\ket{0} = \frac{1}{\sqrt{|\overrightarrow{\alpha}|_1}} \sum_{j=1}^{K} \sqrt{|\alpha_j|}\ket{j},
\label{eq:def_V_lcu}
\end{align}
where $\abs{\vec \alpha}_1 \defeq \sum_{j=1}^K \abs{\alpha_j}$. Moreover, let $O$ be a unitary operation on the joint system such that 
\begin{align}
O=\sum_{j=1}^K \ketbra{j}{j} \otimes \frac{\alpha_j}{\abs{\alpha_j}} U_j.    
\label{eq:def_O_lcu}
\end{align}
Then a direct calculation verifies that 
\begin{align}
\bra{0} (V^\dagger  \otimes I) O (V \otimes I) \ket{0} \ket{\psi} = \dfrac{A \ket{\psi}}{\abs{\vec \alpha}_1}
\end{align}
for all $\ket{\psi} \in \mathcal{H}$. {That is}, after applying $V$ on the ancilla, $O$ on the joint system and $V^\dagger$ on the ancilla, we measure the ancilla in the standard basis, and if the measurement outcome corresponds to $\ket{0}$, the state of the main system becomes $\frac{A\ket{\psi}}{\norm{A\ket{\psi}}}$, and this event happens with probability $\frac{\norm{A\ket{\psi}}^2}{\abs{\vec \alpha}_1^2}$. Figure \ref{fig:lcu_circuit} illustrates the quantum circuit for the LCU method.

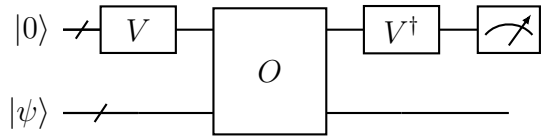
\begin{figure}[htbp]
\centering
\begin{quantikz}
\lstick{$\ket{0}$} &  \gate[wires=1][1cm]{V} \qwbundle{} &
\gate[wires=2][1.5cm]{O} & \gate[wires=1][1cm]{V^\dagger} & \meter{} & 
\\
\lstick{$\ket{\psi}$} & \qw \qwbundle{}  & & \qw & \qw & &
\end{quantikz}
\caption{The LCU circuit for implementing a linear operation $A=\sum_{j=1}^K \alpha_j U_j$ probabilistically. Here $V$ and $O$ are given by Eqs.~\eqref{eq:def_V_lcu} and \eqref{eq:def_O_lcu} respectively.}
\label{fig:lcu_circuit}
\end{figure}

Now we apply the LCU method to Eq.~\eqref{eq:fh_fourier_approx} in which $\alpha_j$ and $t_j$ are given by Eq.~\eqref{eq:f_ftn_dist}. Then we obtain a quantum circuit for implementing the booster operation $f(H)$ approximately and probabilistically. This circuit uses $O(\log{N})$ ancilla qubits and needs to evolve the Hamiltonian $H$ up to time $O(T)$. When the LCU procedure succeeds, it produces the quantum state 
\begin{align}
\dfrac{f_{T,N}(H)\ket{\psi}}{\norm{f_{T,N}(H) \ket{\psi}}} \approx \dfrac{f_T(H)\ket{\psi}}{\norm{f_T(H)\ket{\psi}}} \approx \dfrac{f(H)\ket{\psi}}{\norm{f(H)\ket{\psi}}}
\end{align}
and this event happens with probability
\begin{align}
p_{succ}(f_{T,N})
&\defeq  \dfrac{\bra{\psi}f_{T,N}^\dagger(H)f_{T,N}(H)\ket{\psi}}{\left (\frac{T}{N}\sum_{j=-N}^{N-1} \left |\hat{f}\left ( (j+1/2)\frac{T}{N} \right ) \right |\right )^2} \\
&\approx \dfrac{\bra{\psi}f_T^\dagger(H)f_T(H)\ket{\psi}}{\left (\displaystyle\int_{-T}^T |\hat{f}_T(\xi)| d\xi \right )^2} \\
&\approx \dfrac{\bra{\psi}f^\dagger(H)f(H)\ket{\psi}}{\left (\displaystyle\int_{-\infty}^\infty |\hat{f}(\xi)| d\xi \right )^2}. 
\label{eq:succ_prob}
\end{align}  
In particular, {if $\hat{f}$ is real and even, then $f$ is also real and even. In this case, if $\hat{f}$ is also non-negative}, then the denominator on the RHS of Eq.~\eqref{eq:succ_prob} can be simplified by noting that
\begin{align}
    \displaystyle\int_{-\infty}^\infty \hat{f}(\xi) d\xi    
    = f(0).
\end{align}
Consequently, the success probability of the LCU method in this case is 
\begin{align}
p_{succ}(f_{T,N}) \approx \dfrac{\bra{\psi}f^2(H)\ket{\psi}}{f^2(0)}.
\label{eq:succ_prob2}
\end{align}

\subsection{Designing the booster function}
\label{subsec:design_booster_function}
It remains to design the booster function $f$. To this end, we need certain prior information about the eigenvalues of the Hamiltonian $H$ and the overlaps between the input state $\ket{\psi}$ and the eigenstates of $H$. Specifically, we assume without loss of generality that $H$ has eigenvalues between $0$ and $1$ \footnote{If $H$ has eigenvalues between $\lambda_{\operatorname{min}}$ and $\lambda_{\operatorname{max}}$, then we transform it into $\bar{H}\defeq (H-\lambda_{\operatorname{min}}I)/(\lambda_{\operatorname{max}} - \lambda_{\operatorname{min}})$ which shares the same eigenstates with $H$, and apply our method to $\bar{H}$.}. Furthermore, to facilitate the construction of $f$, we use models for the spectrum of $H$ and the overlaps between $H$'s eigenstates and $\ket{\psi}$. {Specifically}, we {approximate $H$ as}
\begin{align}
    H = {\displaystyle\int_0^1 z \indicator(z) \ket{z}\bra{z} dz},     
    \label{eq:def_H}
\end{align}
where $S \subseteq [0, 1]$ consists of $H$'s eigenvalues, ${\{\ket{z}: z \in S\}}$ are orthonormal states, $\indicator$ is the indicator function of $S$, and {we approximate $\ket{\psi}$ as}
\begin{align}
    \ket{\psi} = {\displaystyle\int_0^1 \psi(z) \ket{z} dz},    
    \label{eq:def_psi}
\end{align}
and 
${q(z)\defeq|\psi(z)|^2}$
is a probability distribution on $[0, 1]$ such that $q(z)=q(z)\indicator(z)$ (i.e. $q(z)=0$ for all $z \not\in S$).

{In practice, we often deal with quantum systems with finite-dimensional Hilbert spaces. For such a system, $H$ has a finite number of eigenvalues and eigenstates, and $\indicator$ should be replaced by a sum of Dirac delta functions, and $q(z)$ should be a discrete distribution. However, if we treat $q(z)$ this way, many properties of $f(H)\ket{\psi}$ that we are interested in will be expressed as functions of a discrete sum over exponentially many terms which is difficult to compute. To mitigate this issue, we approximate $q(z)$ with a continuous distribution so that those properties of $f(H)\ket{\psi}$ can be approximated by functions of a continuous integral which is often easier to calculate. This makes the design of the booster function $f$ more convenient, as will be seen later.} 

Let us consider the impact of the booster operation $f(H)$ on the input state $\ket{\psi}$. Applying $f(H)$ on $\ket{\psi}$ yields the unnormalized state
\begin{align}
    f(H) \ket{\psi} = {\displaystyle\int_0^1 f(z) \psi(z) \ket{z} dz},    
\end{align}
whose squared norm is
\begin{align}
    \norm{f(H) \ket{\psi}}^2 = \bra{\psi} f^\dagger(H) f(H) \ket{\psi} = {\displaystyle\int_0^1 |f(z)|^2  q(z) dz}.    
\end{align}
The expectation of $H$ with respect to $f(H) \ket{\psi}$ is

\begin{align}
    \bra{\psi} f^\dagger(H) H f(H) \ket{\psi} = {\displaystyle\int_0^1 z |f(z)|^2  q(z) dz}.    
\end{align}
Thus, the energy of the normalized state $\frac{f(H) \ket{\psi}}{\norm{f(H) \ket{\psi}}}$ with respect to $H$ is
\begin{align}
    E(f) \defeq \dfrac{\bra{\psi} f^\dagger(H) H f(H) \ket{\psi}}{\bra{\psi} f^\dagger(H)f(H) \ket{\psi}} 
    ={\dfrac{\displaystyle\int_0^1 z |f(z)|^2  q(z) dz}{\displaystyle\int_0^1 |f(z)|^2 q(z) dz}}.
    \label{eq:boosted_energy}
\end{align}
Meanwhile, for arbitrary $\lambda \in [0, 1]$, let 
\begin{align}
    P_{\le \lambda} = {\displaystyle\int_0^\lambda \indicator(z) \ketbra{z}{z} dz}    
\end{align}
be the projection operator onto the subspace spanned by the eigenstates of $H$ whose energies are at most $\lambda$. We are interested in the total overlap between the state $\frac{f(H) \ket{\psi}}{\norm{f(H) \ket{\psi}}}$ and such eigenstates of $H$, i.e. 
\begin{align}
O_{\le \lambda}(f) \defeq \dfrac{\bra{\psi}f^\dagger (H) P_{\le \lambda} f(H) \ket{\psi}}{\bra{\psi}f^\dagger (H)  f(H) \ket{\psi}}
= {\dfrac{\displaystyle\int_0^\lambda  |f(z)|^2  q(z) dz}{\displaystyle\int_0^1 |f(z)|^2 q(z) dz}},    
    \label{eq:boosted_overlap}
\end{align}
where in the second step we use the fact that 
$q(z)=q(z)\indicator(z)$ for all $z \in [0, 1]$. One can see that if $|f|$ decreases monotonically on $[0, 1]$, then $E(f) \le \bra{\psi} H \ket{\psi}$ and $O_{\le \lambda}(f) \ge \bra{\psi} P_{\le \lambda} \ket{\psi}$. {In other words}, $f(H)$ reduces the energy of the input state with respect to $H$ while increasing its overlap with the low-energy eigenstates of $H$.

Ideally, we want $|f|$ to decay as fast as possible on the interval $[0, 1]$, so that the high-energy eigenstates of $H$ get suppressed to a large extent. On the other hand, steep $|f|$ often means that $f \approx f_T$ only for large $T$, and hence it is expensive to implement $f(H)$ using the LCU method. Furthermore, we need to make sure that $f(H)$ can be realized with sufficiently high probability, as it is non-unitary and cannot be realized with certainty. Our goal is to find a function $f$ that achieves the balance among these factors.

Let us first mention a strategy that attempts to directly optimize $f_{T,N}$ (as defined in Eq.~\eqref{eq:def_ftn}). Recall that the LCU method allows us to implement $f_{T,N}(H)$ probabilistically. Given prior information $S$ and $q(x)$ about the Hamiltonian $H$ and the input state $\ket{\psi}$, as well as $T$ (which affects the maximal evolution time of $H$ in the LCU circuit) and $N$ (which determines the number of ancilla qubits in the LCU circuit), we find a high-quality $f_{T,N}$ by solving the following optimization problem: 
\begin{align}
    \max\limits_{\lrcb{\hat{f}_j}}  & \quad O_{\le \lambda}(f_{T,N}), \nonumber \\    
    \textrm{s.t.} & \quad p_{succ}(f_{T,N}) \ge p_0,
    \label{eq:opt_prob}
\end{align}
where the variables are $\{\hat{f}_j\defeq\hat{f}((j+1/2)T/N):~j=-N,\dots,N-1\}$ (which are complex numbers), $f_{T,N}$ is given by Eq.~\eqref{eq:def_ftn}, 
$\lambda$ is an upper bound on the ground state energy of $H$, 
$O_{\le \lambda}$ is given by Eq.~\eqref{eq:boosted_overlap},
$p_{succ}$ is given by Eq.~\eqref{eq:succ_prob}, and $p_0$ is the minimum desired success probability. This is a high-dimensional continuous and constrained optimization problem, so in principle it can be solved by any well-established optimization algorithm on a classical computer. However, given the large number of complex variables and the non-convexity of the objective function and the constraint in these variables, this approach takes a long time and hence is impractical.

To circumvent the above issue and make our scheme practical, we confine $\hat{f}$ to a class of parameterized functions $\hat{f}_{\vec \theta}$ (e.g. Gaussian functions), where $\vec \theta=(\theta_1,\dots,\theta_k)$ are the parameters, and find the optimal solution of Problem \eqref{eq:opt_prob} within this class of functions. This way we only need to solve a low-dimensional optimization problem, which does not take much time on a classical computer. 

Formally, let $f_{\vec \theta}$ be the inverse Fourier transform of $\hat{f}_{\vec \theta}$, i.e. 
\begin{align}
f_{\vec \theta}(x) = \displaystyle\int_{-\infty}^\infty \hat{f}_{\vec \theta}(\xi) e^{i 2\pi x \xi} d\xi.   
\end{align}
To make the booster operation $f_{\vec \theta}(H)$ perform well and easy to implement, we demand that $f_{\vec \theta}$ and $\hat{f}_{\vec \theta}$ satisfy the following requirements:
\begin{itemize}
    \item $|\hat{f}_{\vec \theta}|$ needs to decay exponentially on $\mathbb{R}$, so that we can guarantee that \\
    $\norm{f_{\vec \theta}(H)-f_{\vec \theta;T}(H)} \le \epsilon$ by choosing some $T=O(\log(1/\epsilon))$ as $\epsilon \to 0$, which implies that $f_{\vec \theta}(H)$ can be approximately implemented by a limited-depth circuit.
    \item $|f_{\vec \theta}|$ needs to decrease monotonically on $[0, 1]$, so that $f_{\vec \theta}(H)$ suppresses the high-energy eigenstates of $H$. In fact, we want $|f_{\vec \theta}|$ to decrease as fast as possible on $[0, 1]$, so that $f_{\vec \theta}(H)$ suppresses the high-energy eigenstates of $H$ to a large extent.
\end{itemize}
Here we recommend two classes of functions that meet these conditions: Gaussian and hyperbolic secant. Both of them also have the nice property that the Fourier transform returns a function of the same form as the input function. {That is},
\begin{itemize}
    \item Gaussian function and its Fourier transform: 
    \begin{align}
        f_a(x)=e^{-ax^2} \leftrightarrow  \hat{f}_a(\xi) = \sqrt{\frac{\pi}{a}}e^{-\frac{(\pi \xi)^2}{a}};
    \end{align}
    \item Hyperbolic secant (hsec) function and its Fourier transform: 
    \begin{align}
        f_a(x)=\sech{ax}  
        \leftrightarrow
        \hat{f}_a(\xi) = \frac{\pi}{a}\sech{\frac{\pi^2}{a} \xi}.
    \end{align}
\end{itemize}
By contrast, the exponential function $f_a(x) = e^{-a|x|}$ can be an effective booster choice for suppressing the high-energy eigenstates, but its Fourier transform decays more slowly (i.e. quadratically), resulting in a higher-depth circuit for some fixed $a$.

We acknowledge that the optimal choice of the parameterized functions depends on the specific situation (including $S$ and $q(x)$), and leave it as future work to develop a method for this task.

Once we determine the class of parameterized functions $f_{\vec \theta}$, we find the optimal parameters $\vec \theta$ by solving the following optimization problem:
\begin{align}
    \max\limits_{\vec \theta}  & \quad O_{\le \lambda}(f_{\vec \theta}), \nonumber \\    
    \textrm{s.t.} & \quad p_{succ}(f_{\vec \theta}) \ge p_0, \nonumber \\
    & \displaystyle\int_{-\infty}^{-T} |\hat{f}_{\vec \theta}(\xi)| d\xi + \displaystyle\int_{T}^{\infty} |\hat{f}_{\vec \theta}(\xi)| d\xi \le \delta, \nonumber \\
    & \abs{f_{\vec \theta}(0)} = 1.
    \label{eq:opt_prob2}
\end{align}
{Here $p_{succ}(f_{\vec \theta})$ is defined as 
\begin{align}
p_{succ}(f_{\vec \theta})\defeq \frac{\bra{\psi} f_{\vec \theta}^\dagger(H) f_{\vec \theta}(H) \ket{\psi}}{\lrrb{\int_{-\infty}^{\infty} |\hat{f}_{\theta}(\xi)| d\xi}^2}.    
\end{align} 
In Problem \eqref{eq:opt_prob2}, the second constraint is simply "RHS of Eq.~\eqref{eq:fourier_approx} $\le \delta$", where $\delta>0$ is determined by the method in Appendix \ref{app:choose_delta}. The idea here is that when $\delta$ is small, Eq.~\eqref{eq:fourier_approx} implies that $f_{\vec \theta} \approx f_{T; \vec \theta}$ and hence $O_{\le \lambda}(f_{\vec \theta}) \approx O_{\le \lambda}(f_{T; \vec \theta})$. In other words, this constraint guarantees that the objective function $O_{\le \lambda}(f_{\vec \theta})$ accurately captures the quality of the real booster implemented by the LCU circuit. Appendix \ref{app:choose_delta} gives a detailed relationship between $\delta$ and $|O_{\le \lambda}(f_{\vec \theta})-O_{\le \lambda}(f_{T; \vec \theta})|$ from which an appropriate $\delta$ can be derived. From now on, we will refer to $\delta$ as the \emph{implementation error}, as it is an upper bound on $\norm{f_{\vec \theta}(H) - f_{T; \vec \theta}(H)}$. Meanwhile,} the last constraint $\abs{f_{\vec \theta}(0)} = 1$ is not strictly necessary, but we add it to normalize $f_{\vec \theta}$,  as $\delta$ depends on $|f_{\vec \theta}(0)|=\max_{x \in [0, 1]} \abs{f_{\vec \theta}(x)}$. Then provided that $N$ is sufficiently large, we have $f_{\vec \theta} \approx f_{T; \vec \theta} \approx f_{T,N; \vec \theta}$. It follows that 
$O_{\le \lambda}(f_{\vec \theta}) \approx O_{\le \lambda}(f_{T;\vec \theta}) \approx O_{\le \lambda}(f_{T,N;\vec \theta})$ and $p_{succ}(f_{\vec \theta}) \approx p_{succ}(f_{T;\vec \theta}) \approx p_{succ}(f_{T,N;\vec \theta})$. {In other words}, the actions of $f(H)$, $f_T(H)$ or $f_{T,N}(H)$ on the input state $\ket{\psi}$ will be similar. Note that {if $\hat{f}_{\vec \theta}$ is real and even, then $f_{\vec \theta}$ is also real and even. In this case, if $\hat{f}_{\vec \theta}$ is also non-negative, then $p_{succ}(f_{\vec \theta})$ can be approximated by the RHS of Eq.~\eqref{eq:succ_prob2}}.

\subsection{Example: Gaussian booster}
\label{subsec:gaussian_booster_example}
Now we demonstrate our method on a simple example. Suppose $H=\sum_{j=1}^D \lambda_j \ket{\lambda_j} \bra{\lambda_j}$ has eigenvalues $\lambda_j$'s and eigenstates $\ket{\lambda_j}$'s, where $0 \le \lambda_1 \le \lambda_2 \le \dots \le \lambda_D \le 1$ and the $\lambda_j$'s spread almost uniformly between $0$ and $1$. Moreover, suppose $\ket{\psi}=\sum_{j=1}^D \mu_j \ket{\lambda_j}$ is the state produced by an ansatz circuit, where $\gamma_j \defeq \abs{\mu_j}^2$ decays almost exponentially in $j$. In this case, we set $S=[0, 1]$ and $q(x)={\beta e^{-\beta x}}/{(1-e^{-\beta})}$ for some $\beta>0$ in our framework. 

We confine the booster function $f$ to Gaussian functions $f_a(x)=e^{-ax^2}$ and find the optimal parameter $a$ as follows. By direct calculation, we get
\begin{align}
\bra{\psi} f_a^2(H) \ket{\psi} &= \dfrac{\beta}{1-e^{-\beta}} \displaystyle\int_0^1 
e^{-2ax^2-\beta x} dx \\
&= \dfrac{\sqrt{\pi} \beta e^{\beta^2/(8a)}}{\sqrt{8a} (1-e^{-\beta})} \left [ \operatorname{erf}(\sqrt{2a}+\beta/\sqrt{8a})) - \operatorname{erf}(\beta/\sqrt{8a})\right],
\label{eq:gaussian_psifpsi}
\end{align}
and
\begin{align}
\bra{\psi} f_a(H) P_{\le \lambda} f_a(H) \ket{\psi} &=    
\dfrac{\beta}{1-e^{-\beta}} \displaystyle\int_0^\lambda
e^{-2ax^2-\beta x} dx \\
&= \dfrac{\sqrt{\pi} \beta e^{\beta^2/(8a)}}{\sqrt{8a} (1-e^{-\beta})} \left [ \operatorname{erf}(\lambda \sqrt{2a} + \beta/\sqrt{8a})) - \operatorname{erf}(\beta/\sqrt{8a})\right].
\label{eq:gaussian_psifpfpsi}
\end{align}
In addition, we have
\begin{align}
\displaystyle\int_{-\infty}^{-T} |\hat{f}_{a}(\xi)| d\xi + \displaystyle\int_{T}^{\infty} |\hat{f}_{a}(\xi)| d\xi = 1 - \operatorname{erf}(\pi T / \sqrt{a}).
\label{eq:gaussian_error}
\end{align}
Also, note that $f_a(0)=1$. So we solve the following optimization problem: 
\begin{align}
    \max\limits_{a > 0}  & \quad  \dfrac{\textrm{RHS~of~Eq.}~\eqref{eq:gaussian_psifpfpsi}}{\textrm{RHS~of~Eq.}~\eqref{eq:gaussian_psifpsi}}, \nonumber \\
    \textrm{s.t.} & \quad \textrm{RHS~of~Eq.}~\eqref{eq:gaussian_psifpsi} \ge p_0, \nonumber\\
    & \quad \textrm{RHS~of~Eq.}~\eqref{eq:gaussian_error} \le \delta {,}
    \label{prob:gaussian_booster}
\end{align}
where $\lambda \in \R$ is the best known upper bound on $\lambda_1$, and $\delta >0$ is determined by the method in Appendix \ref{app:choose_delta}. 

{This optimization problem can be solved efficiently on a classical computer as follows. Note that RHS of Eq.~\eqref{eq:gaussian_psifpsi} is a monotonically decreasing function of $a$, and RHS of Eq.~\eqref{eq:gaussian_error} is a monotonically increasing function of $a$. Let $a_1$ and $a_2$ be the solutions of 
RHS~of~Eq.~\eqref{eq:gaussian_psifpsi} $= p_0$ and 
RHS~of~Eq.~\eqref{eq:gaussian_error} $= \delta$, respectively. Both $a_1$ and $a_2$ can be computed (within high accuracy) quickly on a classical computer via binary search. Then the set of feasible solutions to Problem \ref{prob:gaussian_booster} is $(0, \min(a_1, a_2)]$. Meanwhile, we claim that the objective function $\frac{\textrm{RHS~of~Eq.}~\eqref{eq:gaussian_psifpfpsi}}{\textrm{RHS~of~Eq.}~\eqref{eq:gaussian_psifpsi}}$ is a monotonically increasing function of $a$. This implies that the optimal solution to Problem \ref{prob:gaussian_booster} is $a=\min(a_1, a_2)$. To prove this claim, note that 
\begin{align}
    \dfrac{\textrm{RHS~of~Eq.}~\eqref{eq:gaussian_psifpfpsi}}{\textrm{RHS~of~Eq.}~\eqref{eq:gaussian_psifpsi}}
    =\dfrac{\int_0^\lambda e^{-2ax^2-\beta x} dx}{\int_0^1 e^{-2ax^2-\beta x} dx}
    =\dfrac{g(a)}{g(a)+h(a)},
\end{align}
where $g(a)\defeq \int_0^\lambda e^{-2ax^2-\beta x} dx$ and $h(a)\defeq \int_\lambda^1 e^{-2ax^2-\beta x} dx$. For arbitrary $a, \zeta>0$, we have
\begin{align}
    \dfrac{g(a+\zeta)}{g(a)}
    =\dfrac{\int_0^\lambda  e^{-2(a+\zeta)x^2-\beta x} dx}{\int_0^\lambda e^{-2ax^2-\beta x} dx}
    \ge e^{-2\zeta \lambda^2}
    \ge \dfrac{\int_\lambda^1 e^{-2(a+\zeta)x^2-\beta x} dx}{\int_\lambda^1 e^{-2ax^2-\beta x} dx}
    = \dfrac{h(a+\zeta)}{h(a)},
\end{align}
which implies 
\begin{align}
        \dfrac{g(a+\zeta)}{g(a+\zeta)+h(a+\zeta)} \ge \dfrac{g(a)}{g(a)+h(a)},
\end{align}
as claimed.} 

Once we find the optimal parameter $a$, the corresponding booster operation $f_a(H)$ is implemented by the LCU method based on the equation:
\begin{align}
f_{a}(x)= \sqrt{\dfrac{\pi}{a}} \displaystyle\int_{-\infty}^\infty e^{-\frac{(\pi \xi)^2}{a}} e^{2\pi i x \xi} d\xi
        \approx 
        f_{T,N;a}(x)=\dfrac{T}{N} \sqrt{\dfrac{\pi}{a}} \sum_{j=-N}^{N-1} e^{-\frac{(\pi \xi_j)^2}{a}} e^{2\pi i x \xi_j},
\label{eq:gaussian_booster_lcu}
\end{align}
where $\xi_j=(j+1/2)T/N$, and $N$ is sufficiently large so that $f_{a}(x) \approx f_{T;a}(x) \approx f_{T,N;a}(x)$ for all $x \in [0, 1]$. Detailed analysis of Gaussian boosters is given in Appendix \ref{app:analysis_gaussian_booster}. 

{We follow the methodology in Section \ref{subsec:implement_booster_operation} to construct the quantum circuit for implementing $f_{T,N;a}(H)$ probabilistically. For ease of implementation, we choose $N$ to be a power of $2$, i.e. $N=2^n$ for an integer $n \ge 1$. The circuit for $f_{T,N;a}(H)$ is a special instance of the circuit in Figure \ref{fig:lcu_circuit}, where $V$ is an $(n+1)$-qubit unitary operation such that
\begin{align}
    V \ket{0^{n+1}} = \frac{1}{\sqrt{\sum_{k=0}^{2N-1} \alpha_k}}\sum_{k=0}^{2N-1} \sqrt{\alpha_k} \ket{k}
\end{align}
in which $\alpha_k \defeq \frac{T}{N}\hat{f}_a(\xi_{k-N}) = \frac{T}{N}\sqrt{\frac{\pi}{a}}e^{-\frac{(\pi \xi_{k-N})^2}{a}}
=\frac{T}{N}\sqrt{\frac{\pi}{a}}e^{-\frac{\pi^2 (k-N+1/2)^2 T^2}{aN^2}}$, and $O$ is a unitary operation on the joint system such that
\begin{align}
    O = \sum_{k=0}^{2N-1} \ketbra{k}{k} \otimes U^{2k-2N+1}
\end{align}
in which $U \defeq e^{i \pi H T / N}$. Here we have used the fact that $\alpha_k$ is real and positive for each $k$ in defining $V$ and $O$. Note that $O$ is similar to the unitary operation in quantum phase estimation (QPE) and can be implemented as follows:
\begin{align}
    O = U^{-2N+1} O_n O_{n-1} \cdots O_0,
\end{align}
where
\begin{align}
O_p \defeq \ketbra{0}{0} \otimes \mathbb{I}
+
\ketbra{1}{1} \otimes U^{2^{p+1}}
\end{align}
acts on the $p$-th ancilla qubit (which corresponds to $k_p$ assuming $k=\sum_{j=0}^n k_j 2^j$ for $k_0,k_1,\dots,k_n \in \lrcb{0,1}$) and the main system, for $p=0,1,\dots,n$. Figure \ref{fig:gaussian_booster_circuit} illustrates a circuit for implementing $O$ in the case $n=2$.
Note that the final time evolution operation $U^{-2N+1}$ applied to the main system can be dropped because it does not alter the ground state overlap of the output state.
}
\begin{figure}[htbp]
\centering
\begin{quantikz}
& \gate[wires=4][1.5cm]{O} 
& \qw 
& \midstick[4,brackets=none]{=} 
& \qw
& \qw 
& \qw 
& \qw 
& \qw 
& \ctrl{3}
& \qw
& \qw 
\\
&  &\qw &  & \qw & \qw & \qw & \ctrl{2}
& \qw
& \qw 
& \qw 
& \qw 
\\
&  &\qw & & \qw  & \ctrl{1}
& \qw
& \qw & \qw & \qw & \qw
& \qw 
\\[1cm]
& \qwbundle{} & \qw  & & \qwbundle{} & \gate{U^2} & \qw
& \gate{U^{2^2}} & \qw
& \gate{U^{2^3}} & \gate{U^{-(2^3-1)}}
& \qw 
\\
\end{quantikz}
\caption{{The quantum circuit for implementing the unitary operation $O$ for Gaussian boosters in the case $n=2$ has the same structure as the quantum phase estimation (QPE) circuit. The final time evolution $U^{-7}$ on the system register can be removed because it does not alter the weights of the eigenstates.}}
\label{fig:gaussian_booster_circuit}
\end{figure}
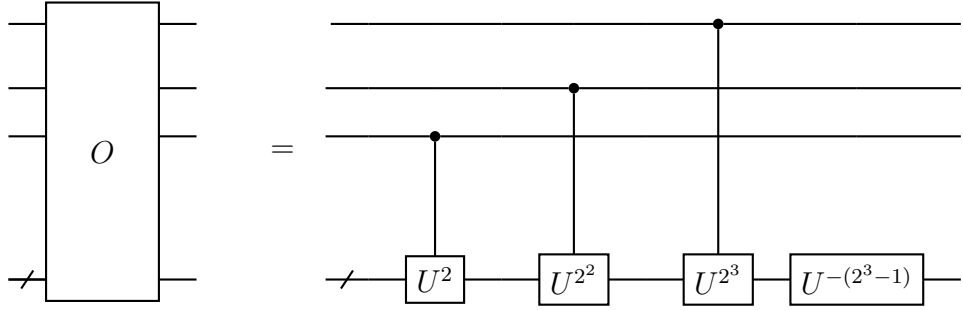

\subsection{Simplified optimization scheme for the Gaussian booster}
\label{subsec:gaussian_simplified}

In the asymptotic regime of large $T$, we expect that the optimized booster will produce an output state with weight concentrated mostly on the ground state and the small remainder on the first excited state.
This allows us to forgo using a weight model of the input state (e.g. the exponential decay model described above) in the optimization.

Furthermore, we will remove the constraints on the success probability $p_0$ and implementation error $\delta$ in the optimization. The success probability is removed so as to focus on the quality of the prepared state instead of the time (in repetitions) needed to implement it.
While the implementation error is ignored as a constraint, we will modify the cost function so that the implementation error plays a direct role in the optimization.

We consider maximizing a lower bound on the ground state overlap that can be evaluated independently of the input state,
\begin{align}
    \quad O_{\le E_{0}}(f_{T;a}) &=\frac{\left|\bra{\psi} f_{T;a}(H)\ket{\lambda_1}\right|^2}{\bra{\psi} f_{T;a}^2(H) \ket{\psi}}\nonumber\\
    &=|\langle b_{T;a} \ket{\lambda_1}|^2\nonumber\\
    &\geq\left(\textup{Re}\left(\langle b_{T;a}\ket{\lambda_1}\right)\right)^2
    \nonumber\\
    &=\left(1-\frac{1}{2}\left\| \ket{\lambda_1}-\ket{b_{T;a}}\right\|^2\right)^2\nonumber\\
    &\geq\left(1-\frac{1}{2}\left(\left\| \ket{\lambda_1}-\ket{b_{a}}\right\|+\left\| \ket{b_{a}}-\ket{b_{T;a}}\right\|\right)^2\right)^2,
    \label{eq:lower_bound_cost_function}    
\end{align}
where we've defined $\ket{b_{T;a}}\equiv\frac{f_{T;a}(H) \ket{\psi}}{\norm{f_{T; a}(H) \ket{\psi}}}$.
{This lower bound to the cost function is maximized by minimizing the sum of the two quantities in Eq. \eqref{eq:lower_bound_cost_function}: the ideal ground state preparation error, $\left\| \ket{\lambda_1}-\ket{b_{a}}\right\|$, and the truncation error, $\left\| \ket{b_{a}}-\ket{b_{T;a}}\right\|$}. As described in 
Appendix \ref{app:asymptotic_performance}, each of these quantities can be upper bounded by expressions involving the Gaussian parameter $a$, the energy gap $\Delta$, the initial ground state overlap $\gamma$, and the truncation level $T$.
{The upper bound for the ideal ground state preparation is:
\begin{align}
    \left\| \ket{\lambda_1}-\ket{b_{a}}\right\| \leq 2 e^{-a \Delta^2} / \gamma,
\end{align}
and the upper bound for the truncation error is:
\begin{align}
\left\| \ket{b_{a}}-\ket{b_{T;a}}\right\|\leq 2\left(1 - \operatorname{erf}(\pi T / \sqrt{a})\right)/\gamma.
\end{align}}
Putting this together, the simplified optimization strategy finds the $a$ which minimizes the sum of these upper bounds on the state preparation error and the truncation error. Noticing that the initial overlap does not impact the optimization, the optimal $a$ is obtained from:
{\begin{align}
    \min_{a}\, \bigg( e^{-a \Delta^2} -  \operatorname{erf}(\pi T / \sqrt{a}) \bigg). \label{eq:simplified_opt}
\end{align}}
In the following section we assess the performance of the resulting boosters through numerical simulation on a specific small molecule example.
For these numerics we use the simplified version of the optimization problem described above, which is suited to the asymptotic performance of the Gaussian booster as described in 
{Appendix} \ref{app:asymptotic_performance}.

\section{Numerical Simulations}
\label{sec:numerics}
\begin{figure}[ht]
    \centering
 \subfloat[Ground state overlap ratio (green) and error in energy (blue) as a function of truncation limit/depth proxy $T$.]{\includegraphics[width=0.45\linewidth]{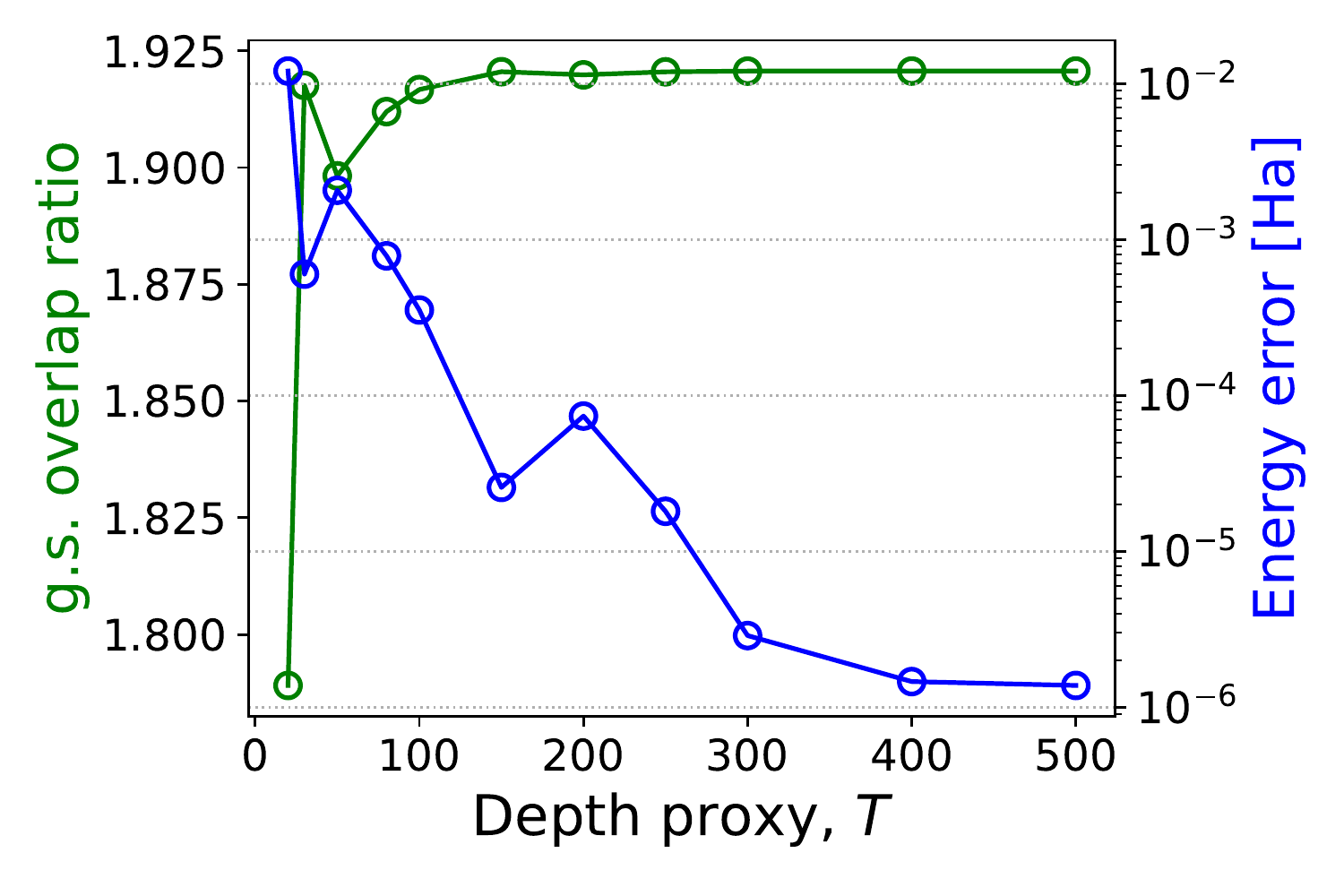}} \hspace{5mm}
 \subfloat[Resulting success probability of the LCU circuits. We show the expected success probability of the ``ideal'' booster (i.e. un-truncated Fourier expansion) in red. In practice, we truncate the Fourier series, and we show the success probability of such ``implemented'' boosters in blue. 
 ]{\includegraphics[width=0.45\linewidth]{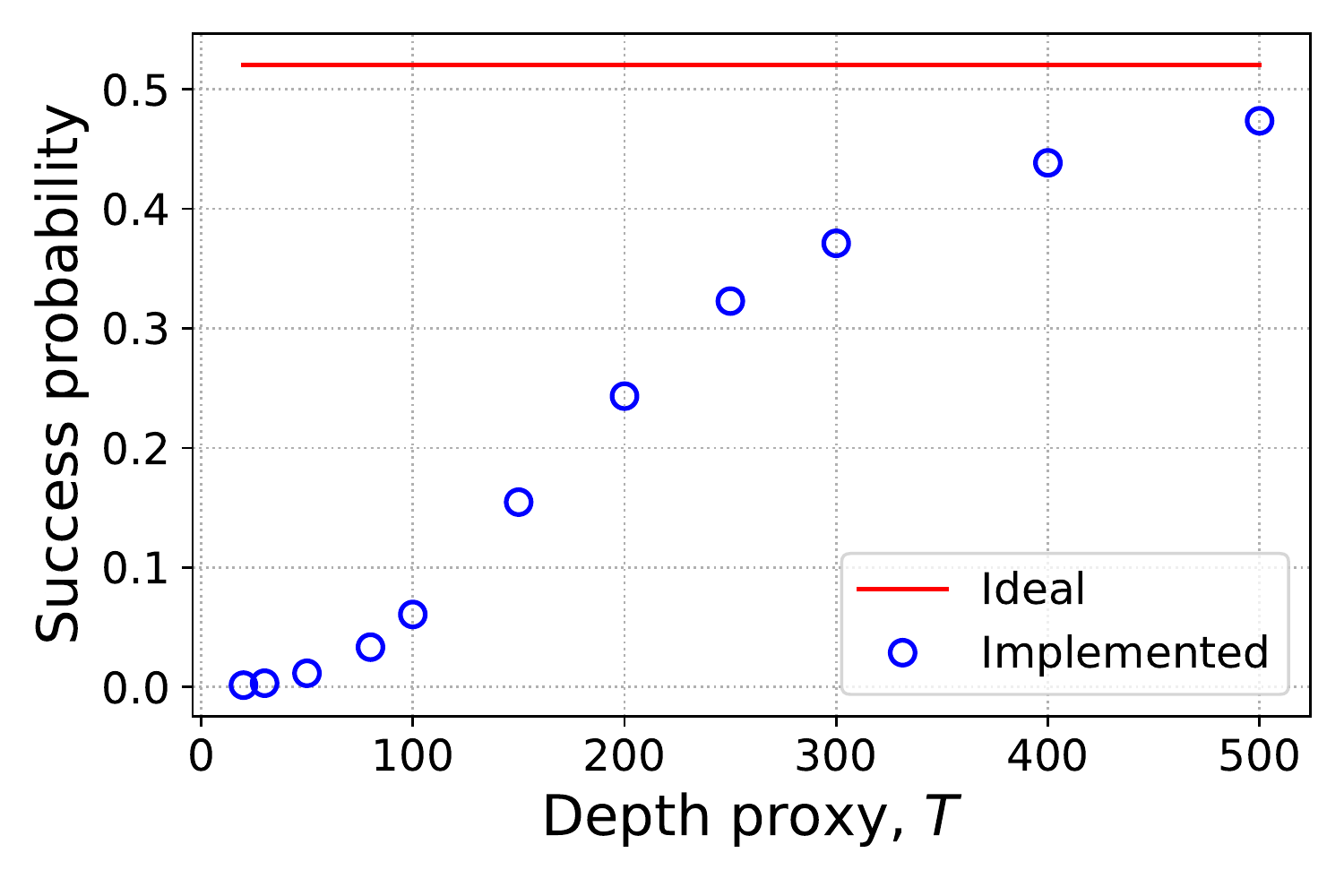}} \\
    \caption{
    Evaluating the performance of the Gaussian booster for preparing the ground state of the $N_2$ molecular system.}
    \label{fig:booster_performance_n2}
\end{figure}

In this section, we provide numerical simulations of the Gaussian booster described in Section \ref{subsec:gaussian_booster_example} using the optimization strategy described in Section \ref{subsec:gaussian_simplified}.
We consider the task of preparing the ground state of the $N_2$ molecular system, considering 6 active electrons in 12 spin-orbitals or qubits in an adapted basis \cite{kottmann2021reducing} at a bond length of $2.0~\angstrom$.
We prepare a mean-field state as the initial state $\ket{\phi_0}$, in which the initial ground state overlap, $\gamma \equiv |\braket{\lambda_1}{\phi_0}|$, is $0.72$. The ground state fidelity is thus $\gamma^2 = 0.52$. Alternatively, there are methods which ensure better overlap with the ground state using short circuit depth (e.g. SPA \cite{kottmann2021optimized}). If we applied boosters to such initializations, compared to boosting the Hartree-Fock state,  we would expect to be able to achieve better overlap.

To determine which Gaussian boosters to use in this setting, as described in Section \ref{subsec:gaussian_simplified} we choose $a$ based on the optimization problem of Eq.~\eqref{eq:simplified_opt}.
In our simulations, we set the spectral gap to $80\%$ of its value (i.e. $0.8 \Delta$) to reflect that in practice, the spectral gap is estimated with limited accuracy.\footnote{The spectral gap, the difference between the lowest and second-lowest energy, is 0.021 Ha before re-scaling.}

Once an optimal $a$ is determined, the performance of the booster can be characterized using several metrics: (1) the ground state overlap ratio, defined as  $\frac{|\langle \lambda_1 | b_{T;a} \rangle|^2}{|\langle \lambda_1 | \psi \rangle|^2}$, where $\ket{b_{T;a}}$ is the normalized state after applying the Gaussian booster and $\ket{\lambda_1}$ is the ground state, (2) the error in energy, i.e. $\langle \phi |H| \phi \rangle - \lambda_1$, and (3) the success probability of implementing the booster.
In Fig.~\ref{fig:booster_performance_n2}, we plot the three quantities as a function of $T$, the truncation level in the truncated Fourier expansion of $f(H)$, or $f_T(H)$. We also consider $T$ as a proxy for the circuit depth as it defines the maximum evolution time in the LCU circuit.
{The circuit implementing the Gaussian booster, as described in the end of Section \ref{subsec:gaussian_booster_example}, involves the same accumulated time evolution duration 
as the operation $c$-$\exp{4\pi i H T}$.
Therefore, measured in terms of number of accumulated $c$-$\exp{2\pi i H}$ operations, the circuit depth is $2T$. For any fixed compilation scheme applied to this controlled time evolution of depth $D$, the circuit depth used in the state preparation booster will be $2TD$.}

From our results, we observe that the booster converts circuit depth into overlap, in which a greater value of $T$ corresponds to a greater ground state overlap ratio or lower error in estimating the ground state energy.
We emphasize that this reliable depth-to-overlap conversion is not a feature of near-term ground state preparation algorithms such as VQE.
Using the error function as our objective, the resulting optimal values of $a$ are on the order of $10^6$ for the values of $T$ we consider. 
At such large values of $a$, for ideal boosters, the ground state overlap ratio saturates to $1/\gamma^2 \approx 1.93$ and the success probability saturates to the initial ground state probability $\gamma^2=0.72^2\approx 0.52$. This latter value is shown in Fig.~\ref{fig:booster_performance_n2}b as a red horizontal line. For the implemented boosters (i.e. with truncation $T$), we observe that for low values of $T$, the success probabilities are significantly lower than the ideal value. This is because at low values of $T$, the actual booster function applied deviates significantly from the ideal booster function of the form $e^{-ax^2}$ as shown in Fig.~\ref{fig:ideal_vs_actual_booster}. As $T$ increases, the success probability approaches the ideal value.  
We also note that the success probability can be further increased by using fixed-point amplitude amplification \cite{yoder2014fixed} at the cost of increased circuit depth.

\begin{figure}[ht]
    \centering
 \subfloat[$T=30$]{\includegraphics[width=0.45\linewidth]{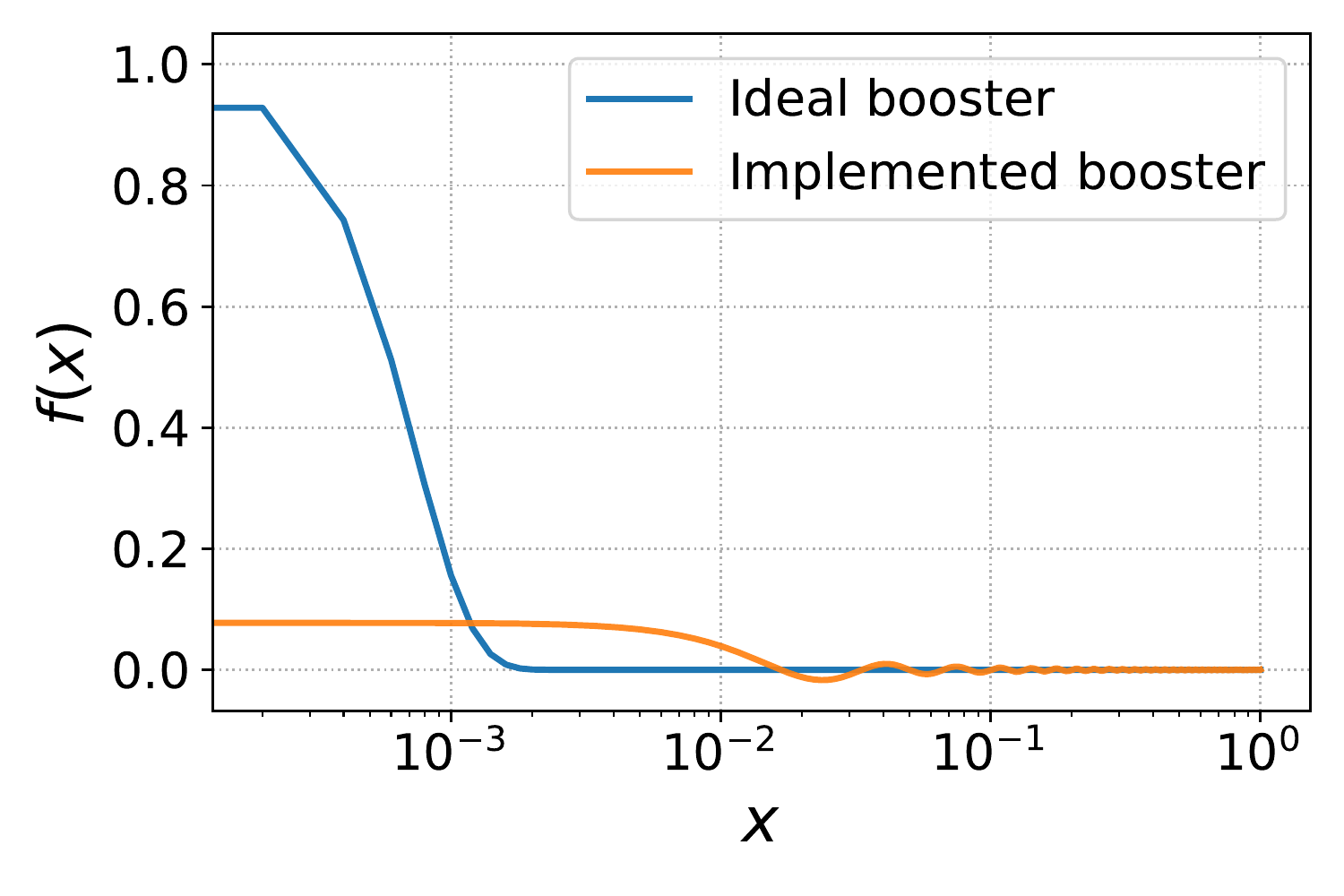}} \hspace{5mm}
 \subfloat[$T=400$]{\includegraphics[width=0.45\linewidth]{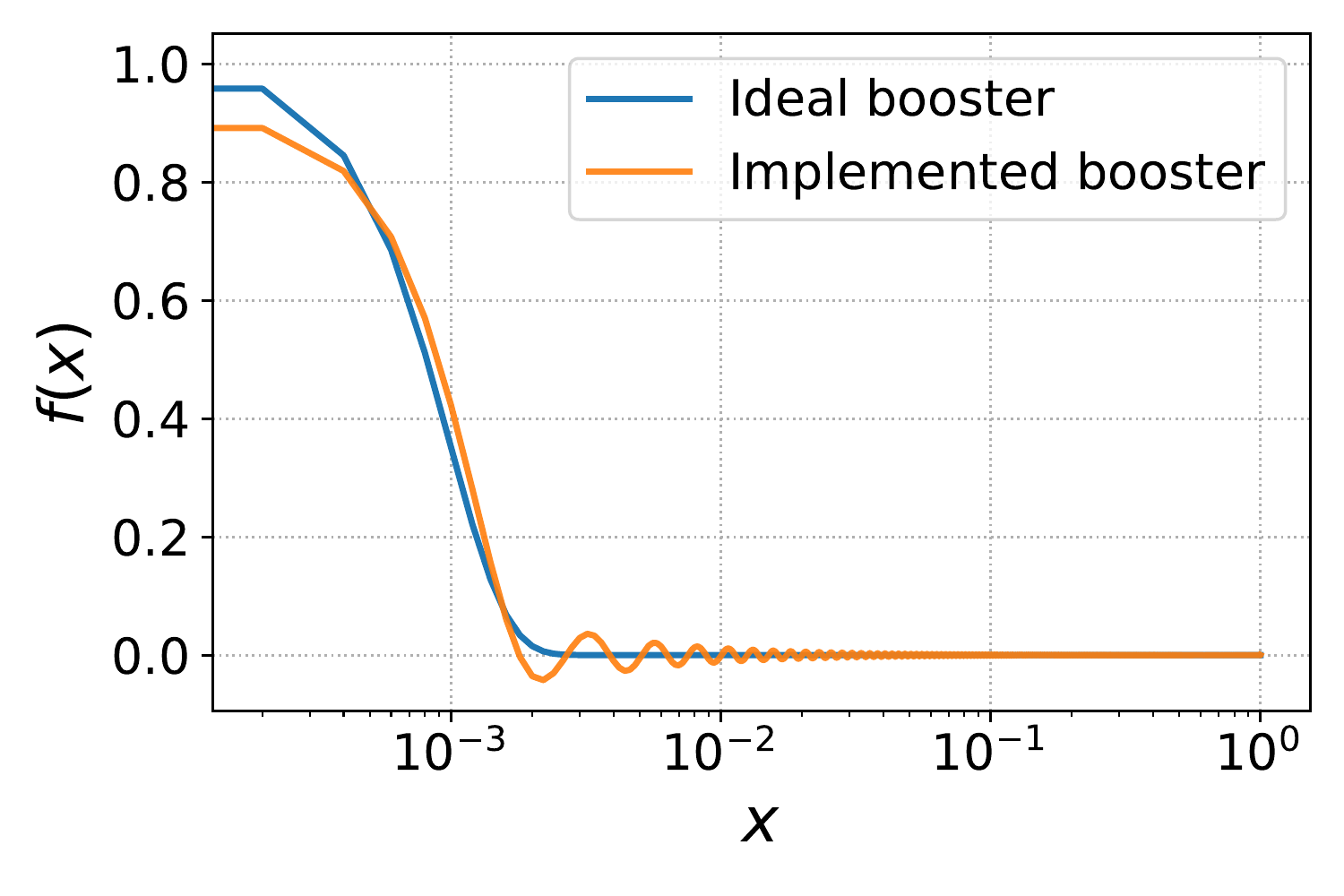}} \\
    \caption{
    Comparing ideal versus implemented booster functions at two different values of $T$.}
    \label{fig:ideal_vs_actual_booster}
\end{figure}

\begin{figure}[ht]
    \centering
    \includegraphics[width=0.5\linewidth]{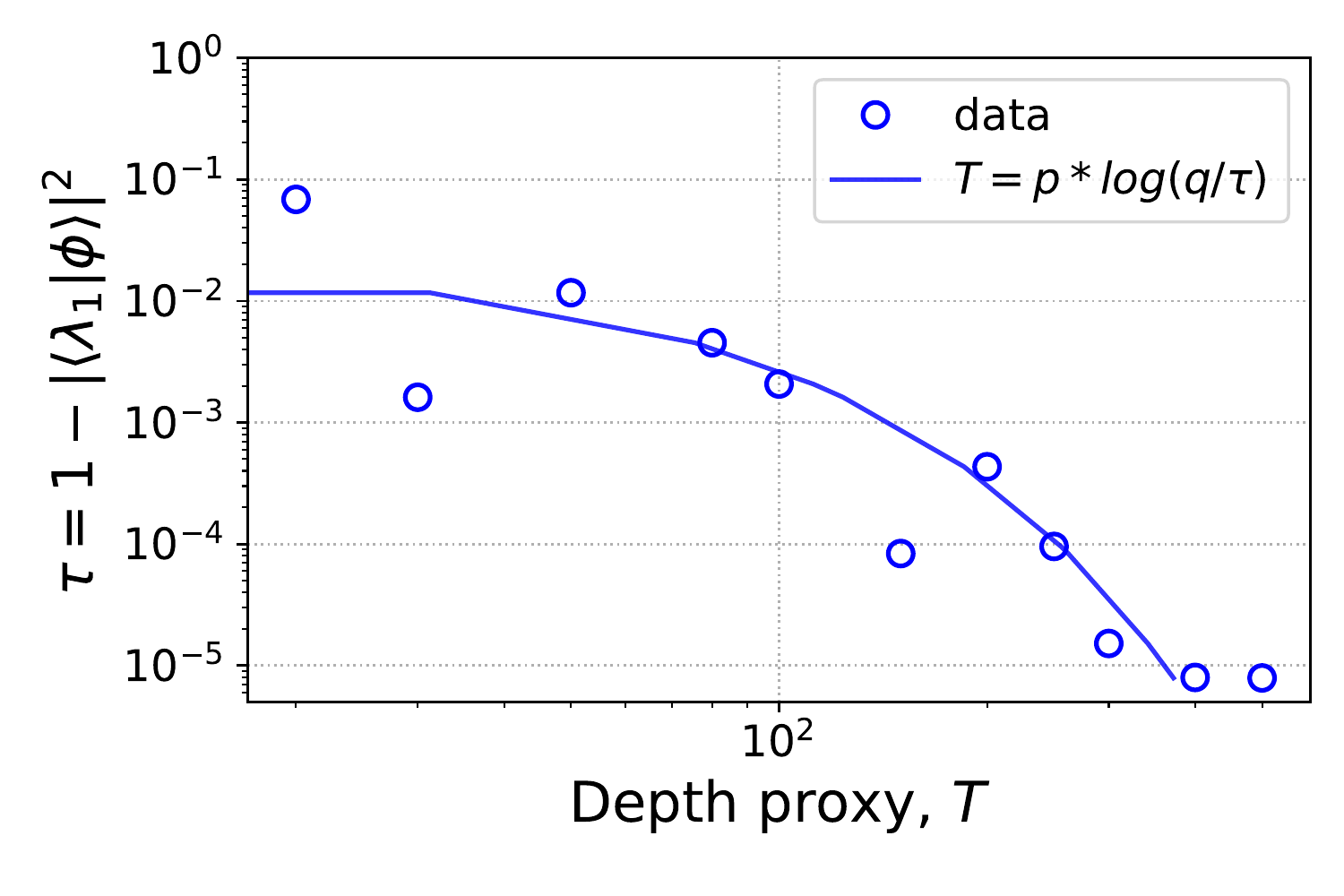}
    \caption{
    Ground state infidelity $\tau$, as a function of the depth proxy $T$. We fit the function of the form 
    $T=p \text{log} (q/\tau)$, with $p \approx 46.72$ and $q \approx 0.023$,
    shown using a solid line.}
    \label{fig:n2_T_scaling}
\end{figure}

Next, we plot the relationship between $T$ and the ground state infidelity $\tau \equiv 1 - |\braket{\lambda_1}{\phi}|^2$ in Fig.~\ref{fig:n2_T_scaling}. 
We show in Appendix~\ref{app:asymptotic_performance} that it is sufficient to choose some $T=\tilde{O}(\log(1/\epsilon))$ where $\epsilon$ and $\tau$ are related by some constant factor. 
We compare the numerical results from this specific example with the analytical predictions by fitting the same functional form in Fig.~\ref{fig:n2_T_scaling} and observe a good fit.
While it might be possible to achieve better performing boosters using a more refined optimization strategy,
our results demonstrate the systematic increase in overlap or decrease in error as $T$ increases.

We briefly provide a back-of-the-envelope calculation comparing the relative cost of {the Gaussian} booster  
to that of the Quantum Phase Estimation algorithm (QPE). {As explained in the end of Section \ref{subsec:gaussian_booster_example}, the circuit operation $O$ used in the Gaussian booster has the same structure as the phase estimation circuit used in QPE.
Therefore, we can compare the depths of these two algorithms in terms of the number of $c$-$e^{2\pi i H}$ applications they accumulate.
Assuming that $H$ has eigenspectrum between 0 and 1, QPE uses $2^{q+1}$ applications of $c$-$e^{2\pi i H}$ to obtain a $q$-bit precision estimate of the phase with high probability.
To compare the circuit depths of the state preparation booster with QPE, we consider the same $N_2$ molecule analyzed above and calculate the number of $c$-$e^{2\pi i H}$ operations needed by QPE to get a chemically-accurate estimate of an eigenenergy.
Chemical accuracy is roughly $10^{-3}$ mHa and the spectrum of the 12-qubit $N_2$ molecule ranges over roughly $10$ mHa.
Therefore, after shifting and rescaling the Hamiltonian such that its spectrum lies within 0 to 1, an energy estimate to within chemical accuracy requires $q=-\log_2(10^{-3}/10)$, or more than 13 bits of precision.
The number of $c$-$e^{2\pi i H}$ needed is therefore roughly 20,000. The Gaussian boosters considered in Figure \ref{fig:booster_performance_n2} use up to $T=500$, which, as described earlier, is roughly the circuit depth of 
$2T=1000$ sequential applications of $c$-$e^{2\pi i H}$.
}

Further investigation is necessary for a more precise comparison of the two methods.
{For example, the number of ancilla qubits used in the Gaussian booster is greater than that used in QPE (c.f. Appendix \ref{app:analysis_gaussian_booster}).}
However, our rough calculation suggests the circuit depth, required for implementing a booster is expected to be a small fraction of the depth needed for QPE.
This gives evidence that state preparation boosters are suited for improving quantum algorithms in the era of early-fault tolerant quantum computers.

\section{Discussion and Outlook}

In this work, we present a framework for improving state preparation with limited quantum circuit depth. The method of \emph{state preparation boosters} reliably converts circuit depth into ground state overlap. This can be used to boost the performance of the variational quantum eigensolver algorithm on early fault-tolerant devices as well as the more far-term method of quantum phase estimation. We leave for future work an analysis of how state preparation boosters can be used to improve the performance of ground state energy estimation. We give a detailed treatment of the specific case of Gaussian boosters, motivated by the performance they achieve while keeping circuit depth relatively low. We carry out simulations of the method on a small (12 qubit) example of the molecule $N_2$. These show that the truncated Gaussian boosters can reliably increase the ground state overlap as circuit depth is increased and they also exhibit the asymptotic scaling derived in Appendix \ref{app:asymptotic_performance}. Below we outline several important future directions to explore in order to better understand and assess the capabilities of boosters. 

First, we have focused on the specific case of Gaussian boosters using controlled-time-evolution circuit primitives (i.e. $c-\exp{iHt}$).
It is worth investigating the performance of alternative booster functions and alternative circuit primitives (e.g. block-encoded Hamiltonians $U = \bigl[ \begin{smallmatrix} H & \cdot \\ \cdot & \cdot\end{smallmatrix}\bigr]$).
Also, it might be possible to refine the booster design process, informed by further numerical investigations.

Along similar lines, it may be worth investigating alternative performance metrics for state preparation.
While boosters reliably convert circuit depth into ground state overlap, particular classes of booster functions may yield better conversion rates than others.
That is, it may be informative to derive a notion of ``efficiency'' of boosters, e.g. amount of overlap gained per unit depth of a booster.
Such an efficiency metric may provide a way to assess and compare general state preparation methods such as VQE or QPE. 

To properly assess the performance of state preparation boosters in practice, we should incorporate the impact of circuit error.
One potential way to account for noise is to consider the circuit depolarizing model as was considered in Ref.~\cite{wang2021minimizing}. 
Using this model, the output state is a mixture of the noiseless output state and the completely mixed state. 
So far we have fixed the circuit depth (or truncation level $T$) of the boosters.
But, by incorporating a model of noise, we might arrive at an optimal choice of booster circuit depth where ideal overlap is balanced with the decay of the quantum coherence.

In addition to circuit depth, another metric associated with boosters is the success probability. 
In the case of ideal boosters (not the truncated boosters analyzed in Fig. \ref{fig:booster_performance_n2}), the more that a booster suppresses high-energy eigenstates, the lower the success probability will be.
In our numerical simulations, we had not designed boosters to accommodate success probability.
But, in practice the success probability will lead to a time overhead of $O(1/p_{success})$ when the booster is used in an algorithm.
One solution, as considered in Section \ref{subsec:design_booster_function}, is to design boosters which yield a success probability above a fixed threshold.
In some cases, this approach may not yield boosters with sufficiently high success probabilities.
In such cases, an alternative strategy is to boost the success probability using amplitude amplification \cite{yoder2014fixed}.
Such techniques cost additional circuit depth and we leave it to future work to analyze the cost-benefit trade-off in the setting of state preparation boosters.
In general, there appears to be a trade-off between the booster depth and success probability. Depending on the constraint on the circuit depth, it may be preferable to apply multiple rounds of shallower boosters in sequence rather than applying a deep booster, with a higher overall success probability, once.
Further studies are necessary to better understand the trade-off and identify when each strategy is preferable. 

Our proposed booster method can be extended to prepare excited states. Using existing near-term methods for estimating low-lying excited states \cite{higgott2019variational,parrish2019quantum,nakanishi2019subspace,lee2019generalized,kottmann2021feasible}, one can obtain an initial state that approximates an excited state. 
Then, with the estimated excited state energy, call it $\tilde{x}$, one can construct a shifted Gaussian booster $f(x) = e^{-a(x-\tilde{x})^2}$ such that the booster function peaks (approximately) at the excited state energy.
We note that many of the steps to implement the booster method for ground state preparation carry over for excited state preparation, including energy re-scaling. 

In many applications, we need to estimate the expectation value of an observable $O$ with respect to the ground state of a Hamiltonian $H$. Our method can be used to improve the accuracy of the estimate of this quantity. Alternatively, one can use the method in Ref.~\cite{zhang2021computing} to estimate this quantity without explicitly preparing the ground state. It would be interesting to know which approach is more efficient in practice.

Boosters provide a reliable depth to overlap conversion, which is lacking in near-term state preparation approaches.
In particular, Gaussian boosters suppress high energy eigenstates while maintaining a relatively low depth, compared to other booster functions. 
Based on our theoretical and empirical analysis, boosters have the potential to become a useful and efficient subroutine for various quantum algorithms, aiding in the quest for quantum advantage. 

\section*{Acknowledgement}
We thank Jakob Kottmann, Micha\l{} St\k{e}ch\l{}y, Athena Caesura, Alex Kunitsa, and Daniel Stilck-Fran\c{c}a for providing valuable feedback to the manuscript. We thank all our colleagues at Zapata Computing for their feedback and support.

\appendix
\section{List of symbols}
\begin{tabular}{cp{0.6\textwidth}}
  $f$ & Booster function \\
  $\lambda$ & An upper bound on ground state energy of $H$ \\
  $T$ & Truncation limit; depth proxy \\
  $N$ & Number of points in discretization in approximation of $f$ \\
  $\Delta$ & Spectral gap \\
  $a$ & Width parameter of the Gaussian booster function \\
  $\delta$ & Target booster implementation error \\
  $p_0$ & Target success probability \\
  $\beta$ & Parameter for modeling initial state weights using a beta distribution \\
  $\gamma$ & Initial overlap with the ground state, $|\langle \lambda_1|\psi \rangle|$ \\
  $\epsilon$ & Lower bound on the error between the ground state and the implemented boosted state
\end{tabular}\\

\section{Bounding the distance between $f_T$ and $f_{T,N}$}
\label{app:dist_ft_ftn}
It is clear from Eqs.~\eqref{eq:def_ft} and \eqref{eq:def_ftn} that $f_{T,N} \to f_T$ as $N \to \infty$. A concrete bound on the convergence speed can be obtained as follows. Assume that $\hat{f}$ is differentiable on $\R$. Let $\xi_j = (j+1/2)T/N$ for $j=-N, -N+1, \dots, N-1$. Then for arbitrary $\xi \in [jT/N, (j+1)T/N]$, the mean value theorem implies that 
\begin{align}
\hat{f}(\xi) = \hat{f}(\xi_j) + \hat{f}'(\eta) (\xi - \xi_j)
\end{align}
for some $\eta$ between $\xi$ and $\xi_j$. {Consequently, we have
\begin{align}
\abs{\hat{f}(\xi) - \hat{f}(\xi_j)} \le \abs{\xi - \xi_j}  \max\limits_{\xi \in [-T, T]} \abs{\hat{f}'(\xi)}, 
\end{align}
and hence for arbitrary $x \in [0, 1]$, }
\begin{align}
\abs{\hat{f}(\xi) e^{i 2\pi x \xi} - \hat{f}(\xi_j) e^{i 2\pi x \xi_j}} 
&\le \abs{\hat{f}(\xi) e^{i 2\pi x \xi} - \hat{f}(\xi_j) e^{i 2\pi x \xi}}  + \abs{\hat{f}(\xi_j) e^{i 2\pi x \xi} - \hat{f}(\xi_j) e^{i 2\pi x \xi_j}} \\
&=\abs{\hat{f}(\xi)  - \hat{f}(\xi_j)}  + \abs{\hat{f}(\xi_j)} \abs{ e^{i 2\pi x \xi} - e^{i 2\pi x \xi_j}} \\
&{\le \abs{\xi-\xi_j} \lrrb{\max_{\xi \in [-T, T]} \abs{\hat{f}'(\xi)} 
+ 2\pi \abs{x} \max_{\xi \in [-T, T]} \abs{\hat{f}(\xi)}}} \\
&\le \frac{T}{N} \lrrb{\max_{\xi \in [-T, T]} \abs{\hat{f}'(\xi)} 
+ 2\pi \max_{\xi \in [-T, T]} \abs{\hat{f}(\xi)}},
\end{align}
{where in the third step we use $\abs{e^{i 2\pi x \xi} -  e^{i 2\pi x \xi_j}} \le 2\pi \abs{x} \abs{\xi-\xi_j}$, and in the last step we use $|\xi - \xi_j| \le T/N$ and $|x|\le 1$}. It follows that
\begin{align}
\abs{f_T(x) - f_{T,N}(x)}
&=\abs{\displaystyle\int_{-T}^T \hat{f}(\xi) e^{i 2\pi x \xi} d\xi
- \frac{T}{N} \sum_{j=-N}^{N-1} \hat{f}\lrrb{\xi_j} e^{i 2\pi x \xi_j}} \\
&{=\abs{\sum_{j=-N}^{N-1}\displaystyle\int_{j T/N}^{(j+1) T/N} \hat{f}(\xi) e^{i 2\pi x \xi} d\xi
- \frac{T}{N} \sum_{j=-N}^{N-1} \hat{f}\lrrb{\xi_j} e^{i 2\pi x \xi_j}} } \\
&\le \sum_{j=-N}^{N-1}
\abs{\displaystyle\int_{j T/N}^{(j+1) T/N} \hat{f}(\xi) e^{i 2\pi x \xi} d\xi
- \frac{T}{N}  \hat{f}\lrrb{\xi_j} e^{i 2\pi x \xi_j}} \\
&{\le \sum_{j=-N}^{N-1}
\displaystyle\int_{j T/N}^{(j+1) T/N} \abs{\hat{f}(\xi) e^{i 2\pi x \xi} - \hat{f}\lrrb{\xi_j} e^{i 2\pi x \xi_j}} d\xi} \\
&{\le \sum_{j=-N}^{N-1}
\displaystyle\int_{j T/N}^{(j+1) T/N} \abs{\hat{f}(\xi) e^{i 2\pi x \xi} - \hat{f}\lrrb{\xi_j} e^{i 2\pi x \xi}} d\xi} \\
&\quad{+\sum_{j=-N}^{N-1}
\displaystyle\int_{j T/N}^{(j+1) T/N} \abs{\hat{f}(\xi_j) e^{i 2\pi x \xi} - \hat{f}\lrrb{\xi_j} e^{i 2\pi x \xi_j}} d\xi}\\
&{= \sum_{j=-N}^{N-1}
\displaystyle\int_{j T/N}^{(j+1) T/N} \abs{\hat{f}(\xi)  - \hat{f}\lrrb{\xi_j} } d\xi} \\
&\quad{+\sum_{j=-N}^{N-1}
\displaystyle\int_{j T/N}^{(j+1) T/N} \abs{\hat{f}(\xi_j)} \abs{e^{i 2\pi x \xi} -  e^{i 2\pi x \xi_j}} d\xi}\\
& \le \frac{2T^2}{N} \lrrb{\max_{\xi \in [-T, T]} \abs{\hat{f}'(\xi)} 
+ 2\pi \max_{\xi \in [-T, T]} \abs{\hat{f}(\xi)}},
\end{align}
{where the third, fourth and fifth steps follow from the triangle inequality, and the seventh step follows from 
\begin{align}
\abs{\hat{f}(\xi) - \hat{f}\lrrb{\xi_j}} \le \abs{\xi - \xi_j} \cdot \max_{\xi \in [-T, T]} \abs{\hat{f}'(\xi) }
\le \frac{T}{N} \cdot \max_{\xi \in [-T, T]} \abs{\hat{f}'(\xi) }    
\end{align}
and 
\begin{align}
\abs{e^{i 2\pi x \xi} -  e^{i 2\pi x \xi_j}} \le 2\pi \abs{x} \abs{\xi-\xi_j} \le \frac{2\pi T}{N}. 
\end{align}
}

{Now let $R \defeq \max_{\xi \in [-T, T]} (|\hat{f}(\xi)|+|\hat{f}'(\xi)|)$. Then we have}
\begin{align}
{\max\limits_{x \in [0,1]}\abs{f_T(x) - f_{T,N}(x)} = O\lrrb{\frac{RT^2}{N}}.}
\end{align}
So for arbitrary $\epsilon>0$, by choosing some $N=O(RT^2/\epsilon)$, we guarantee that $|f_T(x) - f_{T,N}(x)| \le \epsilon$ for all $x \in [0, 1]$. 

Recall that we require that $|\hat{f}|$ decays exponentially on $\R$, so that $\max_{x \in \R}|f(x) - f_{T}(x)| \le \epsilon$ for some $T=O(\log(1/\epsilon))$. In this case, $\hat{f}$ is bounded on $\R$, i.e., there exists a constant $R_1>0$ such that $|\hat{f}(\xi)| \le R_1$ for all $\xi \in \R$. If we require that $\hat{f}'$ is also bounded on $\R$, i.e., there exists a constant $R_2>0$ such that $\abs{\hat{f}'(\xi)} \le R_2$ for all $\xi \in \R$, then $R \le R_1 + R_2 =O(1)$. In this case, it is sufficient to choose some $N=O(T^2/\epsilon)=\tilde{O}(1/\epsilon)$ to ensure that $\max_{x \in [0, 1]}|f_T(x) - f_{T,N}(x)| \le \epsilon$, as claimed. Note that many exponentially-decaying functions (e.g. $\hat{f}(\xi)=e^{-a |\xi|^b}$ for some $a>0$ and $b \ge 1$) have bounded derivatives and hence satisfy this additional requirement.

\section{Choosing an appropriate $\delta$ in Problem \eqref{eq:opt_prob2}}
\label{app:choose_delta}
Here we describe how to set the parameter $\delta$ in Problem \eqref{eq:opt_prob2}. For better readability, we drop the subscript $\vec \theta$ throughout this section. We aim to choose $\delta$ properly so that $f_T$ is close to $f$, which implies that $O_{\le \lambda}(f_T)$ is close to $O_{\le \lambda}(f)$, {i.e.}, the total overlap between the real boosted state and the low-energy eigenstates of $H$ is close to the one for the ideal boosted state.

By Eq.~\eqref{eq:fourier_approx}, we know that the second constraint of Problem \eqref{eq:opt_prob2} ensures that 
\begin{align}
\norm{f(H) - f_T(H)} \le \delta,
\end{align}
which implies that
\begin{align}
\norm{f(H)\ket{\psi} - f_T(H) \ket{\psi}} \le \delta.
\label{eq:dist_fh_fth}
\end{align}
Note that $\norm{f(H) \ket{\psi}} \le 1$ since $|f(x)| \le 1$, for all $x \in [0,1]$. So we have $\norm{f_T(H) \ket{\psi}} \le 1+\delta$.    
{
\begin{lemma}
Suppose $\ket{\psi_1}$ and $\ket{\psi_2}$ are two vectors such that $\norm{\ket{\psi_1}}, \norm{\ket{\psi_2}} \le a$ and $\norm{\ket{\psi_1}-\ket{\psi_2}} \le b$ for some $a, b > 0$, and $A$ is a linear operator such that $\norm{A} \le 1$. Then
\begin{align}
    \abs{\bra{\psi_1} A \ket{\psi_1} - \bra{\psi_2} A \ket{\psi_2}} \le 2 a b.
\end{align}
\label{lem:exp_val_dist}
\end{lemma}
\begin{proof}
Using the triangle inequality, we obtain
\begin{align}
    \abs{\bra{\psi_1} A \ket{\psi_1} - \bra{\psi_2} A \ket{\psi_2}}
    &\le 
\abs{\bra{\psi_1} A \ket{\psi_1} - \bra{\psi_1} A \ket{\psi_2}}
+
\abs{\bra{\psi_1} A \ket{\psi_2} - \bra{\psi_2} A \ket{\psi_2}} \\
& =
    \abs{\bra{\psi_1} A \lrrb{\ket{\psi_1} - \ket{\psi_2}}}
+
\abs{\lrrb{\bra{\psi_1} - \bra{\psi_2}} A \ket{\psi_2}} \\
& \le  a b +  b a \\
& = 2ab.
\end{align}
\end{proof}
}
{It follows from Lemma \ref{lem:exp_val_dist}, Eq.~\eqref{eq:dist_fh_fth},
$f(H)\ket{\psi} \le 1$ and $\norm{f_T(H) \ket{\psi}} \le 1+\delta$ that}
\begin{align}
\abs{\bra{\psi} f^\dagger(H) P_{\le \lambda} f(H) \ket{\psi}
- \bra{\psi} f_T^\dagger(H) P_{\le \lambda} f_T(H) \ket{\psi}} \le 2 \delta (1+\delta),
\label{eq:dist_overlap_plambda}
\end{align}
and
\begin{align}
\abs{\bra{\psi} f^\dagger(H) f(H) \ket{\psi}
- \bra{\psi} f_T^\dagger(H) f_T(H) \ket{\psi}} \le 2 \delta (1+\delta).
\label{eq:dist_overlap}
\end{align}
Meanwhile, the first constraint of Problem \eqref{eq:opt_prob2} ensures that
\begin{align}
{\bra{\psi} f^\dagger(H)  f(H) \ket{\psi} 
= p_{succ}(f) \lrrb{\displaystyle\int_{-\infty}^{\infty} |\hat{f}(\xi)| d\xi}^2
\ge p_0 \lrrb{\displaystyle\int_{-\infty}^{\infty} |\hat{f}(\xi)| d\xi}^2}.
\label{eq:lower_bound_sqared_norm}
\end{align}
Let $\eta$ be a lower bound on $\int_{-\infty}^{\infty} |\hat{f}(\xi)| d\xi$. It can be obtained by, e.g., noting that 
\begin{align}
\int_{-\infty}^{\infty} \abs{\hat{f}(\xi)} d\xi \ge \abs{\int_{-\infty}^{\infty} \hat{f}(\xi) d\xi}=\abs{f(0)}=1,    
\end{align}
where we use the last constraint of Problem \eqref{eq:opt_prob2} in the last step. But one might obtain better bounds by utilizing other information about $f$. Then combining Eqs.~\eqref{eq:dist_overlap_plambda}, \eqref{eq:dist_overlap} and \eqref{eq:lower_bound_sqared_norm} yields
\begin{align}
\abs{O_{\le \lambda}(f) - O_{\le \lambda}(f_T)}
&=\abs{
\dfrac{ \bra{\psi} f^\dagger(H) P_{\le \lambda} f(H) \ket{\psi}}{\bra{\psi} f^\dagger(H) f(H) \ket{\psi}}
-
\dfrac{\bra{\psi} f_T^\dagger(H) P_{\le \lambda} f_T(H) \ket{\psi}}{\bra{\psi} f_T^\dagger(H) f_T(H) \ket{\psi}}
} \\
&\le
\abs{
\dfrac{ \bra{\psi} f^\dagger(H) P_{\le \lambda} f(H) \ket{\psi}}{\bra{\psi} f^\dagger(H) f(H) \ket{\psi}}
-
\dfrac{\bra{\psi} f_T^\dagger(H) P_{\le \lambda} f_T(H) \ket{\psi}}{\bra{\psi} f^\dagger(H) f(H) \ket{\psi}}
} \\
&\quad +
\abs{
\dfrac{ \bra{\psi} f_T^\dagger(H) P_{\le \lambda} f_T(H) \ket{\psi}}{\bra{\psi} f^\dagger(H) f(H) \ket{\psi}}
-
\dfrac{\bra{\psi} f_T^\dagger(H) P_{\le \lambda} f_T(H) \ket{\psi}}{\bra{\psi} f_T^\dagger(H) f_T(H) \ket{\psi}}
} \\
&{= 
\dfrac{\abs{ \bra{\psi} f^\dagger(H) P_{\le \lambda} f(H) \ket{\psi} - \bra{\psi} f_T^\dagger(H) P_{\le \lambda} f_T(H) \ket{\psi} }}{\bra{\psi} f^\dagger(H) f(H) \ket{\psi}}} \\
&\quad 
{+\dfrac{\bra{\psi} f_T^\dagger(H) P_{\le \lambda} f_T(H) \ket{\psi}}{\bra{\psi} f_T^\dagger(H) f_T(H) \ket{\psi}} \cdot
\dfrac{\abs{\bra{\psi} f^\dagger(H) f(H) \ket{\psi} - \bra{\psi} f_T^\dagger(H) f_T(H) \ket{\psi}}}{\bra{\psi} f^\dagger(H) f(H) \ket{\psi}}} \\
&{\le 
\dfrac{\abs{ \bra{\psi} f^\dagger(H) P_{\le \lambda} f(H) \ket{\psi} - \bra{\psi} f_T^\dagger(H) P_{\le \lambda} f_T(H) \ket{\psi} }}{\bra{\psi} f^\dagger(H) f(H) \ket{\psi}}} \\
&\quad 
{+\dfrac{\abs{\bra{\psi} f^\dagger(H) f(H) \ket{\psi} - \bra{\psi} f_T^\dagger(H) f_T(H) \ket{\psi}}}{\bra{\psi} f^\dagger(H) f(H) \ket{\psi}}} \\
&{\le 
\dfrac{2\delta(1+\delta)}{p_0\eta^2}
+\dfrac{2\delta(1+\delta)}{p_0\eta^2}} \\
&{=\dfrac{4\delta(1+\delta)}{p_0\eta^2},}
\end{align}
{where the second step follows from the triangle inequality, the fourth step follows from $\bra{\psi} f_T^\dagger(H) P_{\le \lambda} f_T(H) \ket{\psi} \le \bra{\psi} f_T^\dagger(H) f_T(H) \ket{\psi}$, and the fifth step follows from Eqs.~\eqref{eq:dist_overlap_plambda}, \eqref{eq:dist_overlap} and \eqref{eq:lower_bound_sqared_norm} and the definition of $\eta$.} Consequently, setting {$\delta= (\sqrt{\epsilon p_0 \eta^2 + 1} - 1)/2$} ensures that $\abs{O_{\le \lambda}(f) - O_{\le \lambda}(f_T)} \le \epsilon$. This provides a method for picking an appropriate $\delta$ in Problem \eqref{eq:opt_prob2}.

\section{Detailed analysis of Gaussian boosters}
\label{app:analysis_gaussian_booster}
Here we give detailed analysis of Gaussian boosters. Recall that in Section \ref{subsec:gaussian_booster_example}, we demonstrate how to find the optimal parameter $a$ for the booster function $f_a(x)=e^{-ax^2}$ in the case where $S=[0, 1]$ and $q(x)=\beta e^{-\beta x} / (1 - e^{-\beta})$ for some $\beta>0$.  

By Eqs.~\eqref{eq:fourier_approx} and \eqref{eq:gaussian_error}, we have
\begin{align}
\max_{x \in \R}|f_a(x) - f_{T;a}(x)| 
\le 1 - \operatorname{erf}(\pi T / \sqrt{a})
\le \dfrac{\sqrt{a} e^{-\pi^2 T^2 / a} }{\pi^{3/2} T},   
\label{eq:gaussian_dist_fa_fta}
\end{align}
where in the second step we use the fact that
\begin{align}
1-\operatorname{erf}(x)
=\dfrac{2}{\sqrt{\pi}}
\displaystyle\int_{x}^{\infty} e^{-t^2} dt \le  \dfrac{1}{\sqrt{\pi} x}   \displaystyle\int_{x}^{\infty} 2t e^{-t^2} dt
=\dfrac{e^{-x^2}}{\sqrt{\pi}x}.
\end{align}
This implies that for some $a = \tilde{\Theta}(T^2/\log(1/\delta))$, we have
\begin{align}
\norm{f_a(H)-f_{T; a}(H)} \le \max_{x \in \R}|f_a(x) - f_{T;a}(x)| \le \delta.
\label{eq:gaussian_dist_f_ft}
\end{align}
In this case, we use the results in Appendix \ref{app:dist_ft_ftn} to estimate the smallest $N$ such that
\begin{align}
\norm{f_{T; a}(H)-f_{T,N; a}(H)} \le \max_{x \in [0, 1]}|f_{T;a}(x) - f_{T, N;a}(x)| \le \delta.
\label{eq:gaussian_dist_ft_ftn}
\end{align}
Recall that $\hat{f}_a(\xi)=\sqrt{\frac{\pi}{a}}e^{-\frac{(\pi \xi)^2}{a}}$ is the Fourier transform of $f_a(x)=e^{-ax^2}$. Then we have
\begin{align}
R &\defeq \max_{\xi \in [-T, T]} \lrrb{\abs{\hat{f}_a(\xi)} + 
\abs{\hat{f}_a'(\xi)}}  \\
&= 
\max_{\xi \in [-T, T]} \lrrb{ \sqrt{\frac{\pi}{a}}e^{-\frac{(\pi \xi)^2}{a}}+
\dfrac{2\pi^{3/2}}{a^{3/2}} e^{-\frac{(\pi \xi)^2}{a}} \abs{\xi}}
\\
&{\le 
\max_{\xi \in [-T, T]} \lrrb{ \sqrt{\frac{\pi}{a}}e^{-\frac{(\pi \xi)^2}{a}}}
+
\max_{\xi \in [-T, T]} 
\dfrac{2\pi^{3/2}}{a^{3/2}} e^{-\frac{(\pi \xi)^2}{a}} \abs{\xi}
} \\
&{=\sqrt{\frac{\pi}{a}}+\frac{\sqrt{2\pi}}{\sqrt{e} a}} \\
&{=O\lrrb{\frac{1}{\sqrt{a}}}} \\
&{=\tilde{O}\lrrb{\frac{\sqrt{\log(1/\delta)}}{T}}}.
\end{align}
It follows that Eq.~\eqref{eq:gaussian_dist_ft_ftn} holds for some $N=O(RT^2/\delta)=\tilde{O}(T/\delta)$. 

\subsection{Asymptotic performance of Gaussian boosters}
\label{app:asymptotic_performance}
Now suppose we want the boosted state $\frac{f_{T;a}(H) \ket{\psi}} {\norm{f_{T;a}(H) \ket{\psi}}}$ to be $\epsilon$-close to the ground state of $H$, for arbitrary $\epsilon>0$. How large does $T$ need to be as a function of $\epsilon$? We will show that $T=\tilde{O}(\log(1/\epsilon))$ is sufficient for this purpose. 

Formally, suppose $H=\sum_{j=1}^D \lambda_j \ket{\lambda_j} \bra{\lambda_j}$ has eigenvalues $\lambda_j$'s and eigenstates $\ket{\lambda_j}$'s, where $0 =\lambda_1 < \lambda_2 = \Delta \le \dots \le \lambda_D \le 1$. Moreover, suppose $\ket{\psi}=\sum_{j=1}^D \mu_j \ket{\lambda_j}$ is the state produced by an ansatz circuit such that $ \gamma \defeq \mu_1 > 0$. Then applying $f_a(H)$ or $f_{T;a}(H)$ on $\ket{\psi}$ yields the unnormalized state
$f_a(H) \ket{\psi} = \sum_{j=1}^D f_a(x) \mu_j \ket{\lambda_j}$ or $f_{T;a}(H) \ket{\psi} = \sum_{j=1}^D f_{T;a}(x) \mu_j \ket{\lambda_j}$, respectively. We want to choose appropriate $T$ and $a$ such that
\begin{align}
    \norm{\ket{\lambda_1} - \frac{f_{T;a}(H) \ket{\psi}}{\norm{f_{T;a}(H) \ket{\psi}}}} \le \epsilon
    \label{eq:dist_gs_real_boosted_state}
\end{align}
for given $\epsilon>0$. {By the triangle inequality, we know that Eq.~\eqref{eq:dist_gs_real_boosted_state} holds if we have both}
\begin{align}
    \norm{\ket{\lambda_1} - \frac{f_a(H) \ket{\psi}}{\norm{f_a(H) \ket{\psi}}}} \le \epsilon/2
    \label{eq:dist_gs_ideal_boosted_state}
\end{align}
and
\begin{align}
    \norm{\frac{f_a(H) \ket{\psi}}{\norm{f_a(H) \ket{\psi}}} - \frac{f_{T;a}(H) \ket{\psi}}{\norm{f_{T;a}(H) \ket{\psi}}}} \le \epsilon/2.  
    \label{eq:dist_ideal_real_boosted_states}
\end{align}
It remains to find $T$ and $a$ that satisfy these equations. First, {note that 
\begin{align}
\norm{\gamma \ket{\lambda_1} - f_a(H)\ket{\psi}} = 
\norm{\sum_{j=2}^D \mu_j e^{-a\lambda_j^2} \ket{\lambda_j}}
=
\sqrt{\sum_{j=2}^D |\mu_j|^2 e^{-2a\lambda_j^2}} \le e^{-a\Delta^2}, 
\label{eq:gs_fa_h_psi_dist}
\end{align}
where the last step follows from $\lambda_j \ge \Delta$ for $j=2,3,\dots,D$ and $\sum_{j=2}^D |\mu_j|^2 \le 1$. This implies that
\begin{align}
    \abs{\gamma - \norm{f_a(H)\ket{\psi}}} = 
    \abs{\norm{\gamma \ket{\lambda_1}} - \norm{f_a(H)\ket{\psi}}} \le  e^{-a\Delta^2}.
    \label{eq:gs_fa_h_psi_dist2}
\end{align}
}
As a result, we have
\begin{align}
\norm{\ket{\lambda_1} - \frac{f_{a}(H) \ket{\psi}}{\norm{f_{a}(H) \ket{\psi}}}}
&\le  
\norm{\ket{\lambda_1} - \frac{f_{a}(H) \ket{\psi}}{\gamma}}\\
&\quad +
\norm{\frac{f_{a}(H) \ket{\psi}}{\gamma} - 
\frac{f_{a}(H) \ket{\psi}}{\norm{f_{a}(H) \ket{\psi}}}} \\
&{= \frac{\norm{\gamma \ket{\lambda_1} - f_{a}(H) \ket{\psi}}}{\gamma}} \\
&\quad {+ \frac{\abs{\gamma - \norm{f_a(H)\ket{\psi}}}}{\gamma}}\\
&\le \dfrac{e^{-a\Delta^2}}{\gamma}+\dfrac{e^{-a\Delta^2}}{\gamma} \\
& = \dfrac{2e^{-a\Delta^2}}{\gamma},
\end{align}
{where the first step follows from the triangle inequality, and the third step follows from Eqs.~\eqref{eq:gs_fa_h_psi_dist} and \eqref{eq:gs_fa_h_psi_dist2}.}
Thus, Eq.~\eqref{eq:dist_gs_ideal_boosted_state} holds for some $a=O(\log(1/(\gamma \epsilon))/\Delta^2)$. Meanwhile, Eq.~\eqref{eq:gaussian_dist_fa_fta} implies that
\begin{align}
\norm{f_a(H) \ket{\psi} - f_{T;a}(H) \ket{\psi}} \le 1 - \operatorname{erf}(\pi T / \sqrt{a})
\le \dfrac{\sqrt{a} e^{-\pi^2 T^2 / a} }{\pi^{3/2} T},  
\label{eq:fa_h_psi_fta_h_psi_dist}
\end{align}
{
which in turn implies that
\begin{align}
\abs{\norm{f_a(H) \ket{\psi}} - \norm{f_{T;a}(H) \ket{\psi}}} \le 1 - \operatorname{erf}(\pi T / \sqrt{a})
\le \dfrac{\sqrt{a} e^{-\pi^2 T^2 / a} }{\pi^{3/2} T}.
\label{eq:fa_h_psi_fta_h_psi_dist2}
\end{align}
}
{Then we get}
\begin{align}
    \norm{\frac{f_a(H) \ket{\psi}}{\norm{f_a(H) \ket{\psi}}} - \frac{f_{T;a}(H) \ket{\psi}}{\norm{f_{T;a}(H) \ket{\psi}}}} & \le 
        \norm{\frac{f_a(H) \ket{\psi}}{\norm{f_a(H) \ket{\psi}}} - \frac{f_{T;a}(H) \ket{\psi}}{\norm{f_{a}(H) \ket{\psi}}}} \\
        &\quad +
            \norm{\frac{f_{T;a}(H) \ket{\psi}}{\norm{f_a(H) \ket{\psi}}} - \frac{f_{T;a}(H) \ket{\psi}}{\norm{f_{T;a}(H) \ket{\psi}}}} \\
& {= \frac{\norm{f_a(H) \ket{\psi} - f_{T;a}(H) \ket{\psi}}}{\norm{f_a(H) \ket{\psi}}}} \\
& {\quad + \frac{\abs{\norm{f_a(H) \ket{\psi}} - \norm{f_{T;a}(H) \ket{\psi}}}}{\norm{f_a(H) \ket{\psi}}}}\\
& \le  \dfrac{1 - \operatorname{erf}(\pi T / \sqrt{a})}{\gamma}
+
\dfrac{1 - \operatorname{erf}(\pi T / \sqrt{a})}{\gamma} \\
& = \dfrac{2(1 - \operatorname{erf}(\pi T / \sqrt{a}))}{\gamma} \\
& \le \dfrac{2 \sqrt{a} e^{-\pi^2 T^2 / a} }{\pi^{3/2} T \gamma}, 
\end{align}
{where the first step follows from the triangle inequality, and the third step follows from Eqs.~\eqref{eq:fa_h_psi_fta_h_psi_dist} and \eqref{eq:fa_h_psi_fta_h_psi_dist2}
and $\norm{f_a(H) \ket{\psi}} =\sqrt{\gamma^2 + \sum_{j=2}^D |\mu_j|^2 e^{-2a\lambda_j^2}} \ge \gamma$.} So Eq.~\eqref{eq:dist_ideal_real_boosted_states} holds for some $T=\tilde{O}(\sqrt{a \log(1/(\gamma\epsilon))})=\tilde{O}(\log(1/(\gamma \epsilon))/\Delta)$. 
Overall, by choosing some $T=\tilde{O}(\log(1/(\gamma \epsilon))/\Delta)$ and $a=O(\log(1/(\gamma \epsilon))/\Delta^2)$, we ensure that Eq.~\eqref{eq:dist_gs_real_boosted_state} holds, as claimed.

\end{document}